\documentclass{article}
\usepackage[T1]{fontenc}
\usepackage[utf8]{inputenc}
\usepackage{csquotes}
\usepackage{url}
\usepackage[hidelinks]{hyperref}
\usepackage[centertags]{amsmath} \usepackage{amssymb} 
\usepackage{amsthm} \usepackage{amsfonts} \usepackage[margin=3.5cm]{geometry} 
\usepackage[round,authoryear]{natbib}
\usepackage{graphicx} \usepackage{tikz} \usetikzlibrary{arrows.meta,positioning,calc,fit} \usepackage{subcaption} \usepackage{authblk} \usepackage{xcolor} \usepackage{algorithm}
\usepackage{algpseudocode}
\usepackage{setspace}
\usepackage{placeins}
\usepackage{booktabs}
\usepackage{tabularx}

\newtheorem{theorem}{Theorem} \newtheorem{lemma}[theorem]{Lemma} \newtheorem{corollary}[theorem]{Corollary} \theoremstyle{definition}  \theoremstyle{remark} 

\usepackage{float}\usepackage{textcase}

\hypersetup{
  pdftitle={Regime-Switching Stochastic Epidemics: Joint Analysis of Extinction Time and Outbreak Size},
  pdfauthor={Vasileios E. Papageorgiou, Irene Votsi, Samis Trevezas},
  pdfkeywords={finite-population epidemic, continuous-time Markov chain, Markovian regime switching, state-dependent intervention, extinction time, final size}
}
\definecolor{mypink1}{rgb}{0.858, 0.188, 0.478}

\begin{document}

\onehalfspacing

\title{Evaluating the Impact of Epidemic Control via State-Dependent Markovian Switching Modeling}

\author[1]{Vasileios E. Papageorgiou%
\thanks{Corresponding author. Email: \texttt{vpapageor@math.uoa.gr};
ORCID: 0000-0002-8131-3484.}}

\author[2]{Irene Votsi%
\thanks{Email: \texttt{Eirini.Votsi@univ-lorraine.fr};
ORCID: 0000-0002-7397-2025.}}

\author[1,3]{Samis Trevezas%
\thanks{Email: \texttt{strevezas@math.uoa.gr};
ORCID: 0000-0003-2262-8299.}}

\affil[1]{Department of Mathematics, National and Kapodistrian University
of Athens, Panepistimiopolis, Athens, 15784, Greece.}

\affil[2]{LIEC, CNRS, Universit\'e de Lorraine,
F-57000 Metz, France.}

\affil[3]{MICS Laboratory, CentraleSup\'elec,
Universit\'e Paris-Saclay, 3 Rue Joliot Curie,
Gif-sur-Yvette, 91190, France.}
\date{} \maketitle

\begin{abstract}
\noindent We develop an exact finite-population stochastic framework for SIR epidemics evolving under Markovian switching between intervention
regimes. The epidemic state is augmented by a finite phase component,
allowing transmission, recovery, and direct immunity-acquisition rates
to depend on the active regime. Phase-transition intensities
may depend on the current epidemic state, so that policy
escalation can react to the number of infectious individuals.
Exploiting the monotonicity of the susceptible compartment, we derive
level-wise recursions for the joint Laplace--Stieltjes transform and
probability generating function of the extinction time and the number
of infections generated before extinction. These recursions yield the
infection-count distribution, conditional extinction-time transforms,
and mixed moments linking epidemic duration and infection burden,
while replacing a large global linear system with small phase-level
solves.

The framework is illustrated using weekly mpox incidence data from
Luxembourg. A baseline one-phase SIR model is calibrated by maximum likelihood under a Poisson observation 
model. The calibrated baseline is then used for conditional comparisons
of fixed control regimes, early versus delayed strict intervention,
vaccination-supported control, and state-dependent escalation. The
results show how switching mechanisms affect both the total number of
infected individuals and the extinction time, including their
dispersion. Since the switching mechanisms are specified
rather than estimated from the intervention history, the results are
conditional model-based comparisons rather than
estimates of the historical effects of interventions in Luxembourg.

\\


\end{abstract}


\section{Introduction} \label{sec1}
Recent epidemic crises have stimulated a broad range of methods for describing and forecasting disease transmission \citep{Papageorgiou2025r}. In finite populations, however, the comparison of intervention mechanisms requires more than a projected mean trajectory: both the number of individuals ultimately infected and the time until extinction are random. A useful stochastic framework should therefore characterize infection burden and epidemic duration jointly, while allowing intervention intensity to change as the outbreak evolves.

Foundational deterministic models typically build on the SIS or SIR structure \citep{Cooper2020}, with the aim of capturing the core dynamics of infection and recovery. 
Well-known generalizations include the SEIR
and SEIRS models
, while extensions may include quarantined
individuals, vaccination, or open populations characterized by
migration
\citep{YALADANDA2022101052,Kiss2024,Papageorgiou2026d}. Epidemic models have also been employed to describe social dynamics, with emphasis on the transmission of criminality and radicalization \citep{Santoprete2019,Sooknanan2023,Papageorgiou2026g}. However, these formulations are predominantly deterministic and therefore cannot fully represent the intrinsic uncertainty in epidemic processes.

A large body of work embeds interventions as time-dependent controls in differential equation models and derives candidate policies via Pontryagin’s maximum principle, typically validated through numerical simulations \citep{Zamir2021}. Deterministic epidemic models are also frequently combined with local or global sensitivity analyses to identify influential transmission and intervention parameters before formulating control strategies \citep{Lu2023b}. Applied work estimates intervention impact by calibrating deterministic transmission models to data and running counterfactual simulations \citep{Flountzi2022}.
These approaches are valuable for designing and comparing prespecified controls, but they generally do not provide the joint finite-population distribution of epidemic duration and infection burden when the intervention regime changes randomly or in response to the epidemic state.

Stochastic epidemic models are especially informative for small populations or for settings where outcomes are sensitive to demographic, environmental, or transmission heterogeneity \citep{Allen2017,Britton2010}. Numerous studies employ low-dimensional continuous-time Markov chain (CTMC) models to efficiently describe the dynamics of epidemics. Examples include the maximum number of infections in heterogeneous-contact SIS systems \citep{Economou2015} and in standard SIS models \citep{Artalejo2010b}, as well as extreme-event behavior in a two-strain SIS model \citep{Almaraz2019}. Further work derives stochastic descriptors that provide exact measures of disease spread \citep{Artalejo2013}, while  \citet{Gmez-Corral2023} extend the exact reproduction number to SIS with vertical transmission, separating contributions from contacts and infected newborns in the resulting distribution. Matrix-form expressions have been derived for several stochastic descriptors, such as the infection time of a susceptible individual, the total numbers of infections and deaths, and the maximum and total number of hospitalizations, considering stochastic SIRD and SIHRD schemes \citep{Papageorgiou2024b,Papageorgiou2023a,Papageorgiou2024d}. Recent work has also combined particle filtering with stochastic compartmental models to dynamically estimate epidemic descriptors, including in the SPIR setting \citep{Papageorgiou2026SPIR}. The dynamics of fundamental epidemic models in random environments have also been studied \citep{Artalejo2013b}, and discrete-time counterparts of these CTMC formulations are also available \citep{Gamboa2018, Papageorgiou2025c, Gmez-Corral2021}.
Extinction time and final epidemic size are especially informative in this setting: the former measures how long transmission persists, whereas the latter measures cumulative burden. Their dependence may distinguish intervention mechanisms that appear similar when only marginal means are compared.

A closely related line of work concerns finite-population SIR epidemics with Markov-modulated event mechanisms. \citet{ArtalejoGomezCorral2010} proposed a state-dependent Markov-modulated mechanism in which an auxiliary finite-state phase process governs event occurrences. This construction allows non-Poisson event streams and dependence between successive inter-event times to be represented while preserving a finite-dimensional Markov description. \citet{AlmarazGomezCorral2018} incorporated this mechanism into a finite-population SIR epidemic, using Markov-modulated infection and removal processes to accommodate non-exponential waiting-time distributions and different correlation structures. Their analysis considers outbreak duration, final epidemic size, and the number of secondary infections. In that framework, the modulating phases are introduced primarily to enrich the temporal structure of infection and removal events. The present construction has a different interpretation and objective: the phase $J(t)$ represents the active intervention regime, simultaneously determines the transmission, recovery, and direct immunity-acquisition intensities, and may switch at rates that depend explicitly on the current epidemic state $(S(t),I(t))$. Thus, rather than using modulation mainly to generate correlated event times, we use it to represent random or epidemic-responsive intervention changes and to derive the joint distribution of extinction time and infection burden through a level-wise recursion.

The contributions are twofold. First, we formulate a finite-population stochastic SIR model in which the intervention phase jointly modulates transmission, recovery, and direct immunity acquisition, while the phase-transition intensities may depend explicitly on the current numbers of susceptible and infectious individuals. Relative to the preceding Markov-modulated formulations, the phase is interpreted as an intervention regime and is used to represent uncertain policy escalation, relaxation, and vaccination-supported immunity within a single construction. Second, we develop an exact recursive method for the joint distributional analysis of extinction time and infection burden. By exploiting the monotone decrease of the susceptible population, the proposed recursions avoid a large global linear system and reduce the computation to small phase-level linear systems.

The resulting recursive scheme yields the extinction-time transform and moments, the distributions of new and total infections, conditional extinction-time quantities, and mixed descriptors linking duration and burden. It therefore permits full distributional comparisons of intervention mechanisms, including cases in which marginal means are similar but uncertainty or dependence differs.

To illustrate the method, we use weekly mpox incidence data from Luxembourg to calibrate a baseline one-phase SIR model and then evaluate alternative intervention mechanisms conditionally on the fitted parameters. The numerical scenarios compare fixed control regimes, early versus delayed strict intervention, vaccination-supported control, and state-dependent escalation. Because all scenarios share the same calibrated epidemic parameters, differences in their outcome distributions can be attributed, within the model, to the intervention mechanism specified in each experiment. The switching intensities and intervention effects are specified rather than estimated from the historical policy process. The numerical study should therefore be interpreted as a conditional comparison of model mechanisms, not as a causal evaluation of interventions implemented in Luxembourg.

The remainder of the paper is organized as follows. 
Section~\ref{sec2} presents the state-dependent Markovian regime-switching SIR model and the recursive formulas for the joint analysis of extinction time and infection burden. 
Section~\ref{sec3} applies the method to weekly mpox incidence data from Luxembourg and compares the proposed intervention scenarios. 
Section~\ref{sec4} discusses the main methodological findings, limitations, and possible extensions. 
Appendix~\ref{app:proofs} contains the proofs of the main results.


\section{Model formulation and distributional analysis}
\label{sec2}

This section develops the finite-state continuous-time Markov chain and its distributional analysis. The active intervention regime modulates transmission, recovery, and direct immunity acquisition, while transitions between regimes may depend on the current numbers of susceptible and infectious individuals. The central result is a level-wise recursion for the joint Laplace--Stieltjes transform--probability generating function of the extinction time and the number of infections generated before extinction. By exploiting the non-increasing behavior of the susceptible population, the recursion replaces a global linear system over the full transient state space with a sequence of smaller phase-level systems. The same construction yields the infection-count distribution, extinction-time transforms and moments, conditional quantities, and mixed descriptors linking epidemic duration and infection burden. These quantities are obtained exactly up to numerical linear-system solution; only the continuous extinction-time densities used in Section~\ref{sec3} require numerical Laplace-transform inversion. State-independent phase transitions are obtained as a special case.

\subsection{Regime-switching SIR process}
\label{subsec:model}

We first give the finite-state construction of the model. The usual SIR population balance is kept, but the epidemic process is augmented by a phase coordinate describing the active intervention regime. The phase modifies the epidemic transition intensities, while transitions between phases may depend on the current epidemic state. Thus, intervention changes can be specified so as to react to the observed epidemic burden, while the pair formed by the epidemic state and the intervention phase remains Markovian.

We work with a closed population of size $N$ and with $P$ possible intervention phases. At time $t\geq0$, $S(t)$, $I(t)$ and $R(t)$ denote the numbers of susceptible, infectious and removed or effectively immune individuals, respectively. Hence
\begin{equation*}
S(t)+I(t)+R(t)=N,\qquad t\geq0.
\end{equation*}
The removed class contains recovered individuals, as well as susceptible individuals who have acquired perfect immunity, for instance through vaccination.

To avoid confusion between the intervention-regime process, the number $P$ of phases, and probability notation, we denote the active intervention phase by $J(t)$. It takes values in $S_p=\{p:1\leq p\leq P\}$. Since $R(t)=N-S(t)-I(t)$, the process can be written as
\begin{equation*}
\mathcal{X}=\{X(t)=(S(t),I(t),J(t)):t\geq0\}.
\end{equation*}
We assume that $X(0)=(N-i_0,i_0,p_0)$ and $R(0)=0$, where $1\leq i_0\leq N$ and $1\leq p_0\leq P$. The quantities derived below are nevertheless defined for every admissible transient state.

For $p'\neq p$ and for each admissible epidemic state $(s,i)$, let $\lambda_{p,p'}(s,i)\geq0$ be the transition intensity from phase $p$ to phase $p'$ when the epidemic state is $(s,i)$. We assume that these intensities are finite. Set
\begin{equation*}
\lambda_{p,p}(s,i)=-\sum_{\substack{p'=1\\p'\neq p}}^{P}\lambda_{p,p'}(s,i),\qquad 1\leq p\leq P,
\end{equation*}
and write
\begin{equation*}
\mathbf{Q}_P(s,i)=\bigl(\lambda_{p,p'}(s,i)\bigr)_{1\leq p,p'\leq P}.
\end{equation*}
For fixed $(s,i)$, $\mathbf{Q}_P(s,i)$ is the generator governing instantaneous changes of the intervention phase. Since it may vary with $(s,i)$, the phase coordinate alone need not be Markovian; the joint process $\mathcal{X}$ is Markovian. The state-independent model is obtained by taking
\begin{equation*}
\lambda_{p,p'}(s,i)=q(p,p'),\qquad p'\neq p,
\end{equation*}
for all admissible $(s,i)$.

The epidemic part of the state space is $S_s=\{(s,i):0\leq s\leq N-i_0,\;0\leq i\leq N-s\}$, and the full state space is
\begin{equation*}
S=S_s\times S_p=\{(s,i,p):0\leq s\leq N-i_0,\;0\leq i\leq N-s,\;1\leq p\leq P\},
\end{equation*}
with $|S|=P(N-i_0+1)(N+i_0+2)/2$. For the level construction, put
\begin{equation*}
L(s,i)=\{(s,i,p):1\leq p\leq P\},\qquad L(s)=\bigcup_{i=0}^{N-s}L(s,i),\qquad S=\bigcup_{s=0}^{N-i_0}L(s).
\end{equation*}
An infection or a direct immunity-acquisition event moves the process from level $L(s)$ to level $L(s-1)$, whereas a recovery or a phase change leaves $s$ unchanged. Thus, $S(t)$ is non-increasing.

The entries of the infinitesimal generator $\mathbf{Q}=(q_{(s,i,p),(s',i',p')})$ are
\begin{equation}
\label{eq:generator}
q_{(s,i,p),(s',i',p')}=
\begin{cases}
\dfrac{b_p}{N}\,s\,i,&(s',i',p')=(s-1,i+1,p),\\[1ex]
\gamma_p\,i,&(s',i',p')=(s,i-1,p),\\[1ex]
\psi_p\,s,&(s',i',p')=(s-1,i,p),\\[1ex]
\lambda_{p,p'}(s,i),&(s',i',p')=(s,i,p'),\quad p'\neq p,\\[1ex]
-q_{(s,i,p)},&(s',i',p')=(s,i,p),\\[1ex]
0,&\text{otherwise},
\end{cases}
\end{equation}
where
\begin{equation*}
q_{(s,i,p)}=\dfrac{b_p}{N}\,s\,i+\gamma_p\,i+\psi_p\,s+\sum_{\substack{p'=1\\p'\neq p}}^{P}\lambda_{p,p'}(s,i),
\end{equation*}
is the total transition intensity out of $(s,i,p)$. A transition in~\eqref{eq:generator} is omitted when its stated destination does not belong to $S$.

For phase $p$, $b_p\geq0$ is the transmission coefficient, $\gamma_p\geq0$ the recovery intensity per infectious individual and $\psi_p\geq0$ the direct immunity-acquisition intensity per susceptible individual. Thus, in state $(s,i,p)$, the total infection, recovery and direct immunity-acquisition intensities are $b_p\,s\,i/N$, $\gamma_p\,i$ and $\psi_p\,s$, respectively. If $\psi_p$ is interpreted as a vaccination intensity, protection is assumed to be immediate and to last until the end of the outbreak. Imperfect or waning protection is not represented by this SIR structure.

The model does not impose an ordering on $b_1,\ldots,b_P$. In applications, a phase representing stronger contact reduction may be assigned a smaller value of $b_p$. The functions $\lambda_{p,p'}(s,i)$ specify how interventions are introduced, intensified or relaxed as the epidemic state changes. Setting selected functions $\lambda_{p,p'}$ equal to zero gives one-way or prohibited phase changes. Figure~\ref{fig:scheme} shows a three-phase example; the labels $q(p,p')$ correspond to the state-independent special case.

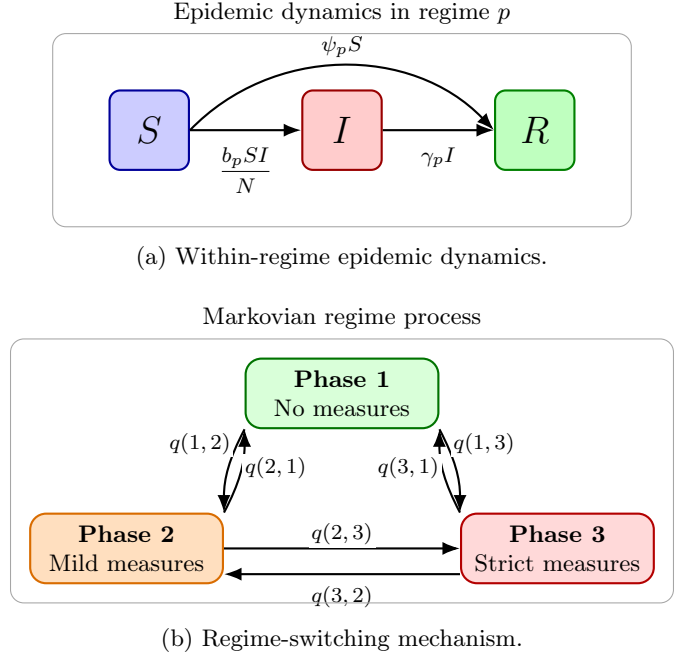
\begin{figure}[!htbp]
\centering
\captionsetup[subfigure]{justification=centering}

\begin{subfigure}[t]{0.7\textwidth}
\centering
\begin{tikzpicture}[
    >=Latex,
    font=\small,
    comp/.style={
        rectangle,
        rounded corners=4pt,
        minimum width=1.05cm,
        minimum height=1.05cm,
        align=center,
        thick,
        font=\Large\bfseries
    },
    flow/.style={-Latex, thick},
    lab/.style={font=\footnotesize, fill=white, inner sep=1pt}
]
\node[comp, draw=blue!60!black, fill=blue!20] (S) at (0,0) {$S$};
\node[comp, draw=red!60!black, fill=red!20] (I) at (2.55,0) {$I$};
\node[comp, draw=green!50!black, fill=green!25] (R) at (5.10,0) {$R$};
\draw[flow] (S) -- node[below=5pt, lab] {$\dfrac{b_pSI}{N}$} (I);
\draw[flow] (I) -- node[below=5pt, lab] {$\gamma_p I$} (R);
\draw[flow] (S.east) to[out=46,in=134] node[pos=0.50, above=2pt, lab] {$\psi_p S$} (R.west);
\node[
    draw=black!35,
    rounded corners=5pt,
    inner sep=21pt,
    fit=(S)(I)(R),
    label={[font=\small]above:{Epidemic dynamics in regime $p$}}
] {};
\end{tikzpicture}
\caption{Within-regime epidemic dynamics.}
\end{subfigure}

\vspace{0.9em}

\begin{subfigure}[t]{0.64\textwidth}
\centering
\begin{tikzpicture}[
    >=Latex,
    font=\small,
    phase/.style={
        rectangle,
        rounded corners=6pt,
        minimum width=2.55cm,
        minimum height=0.92cm,
        align=center,
        thick
    },
    switch/.style={-Latex, thick},
    lab/.style={font=\footnotesize, fill=white, inner sep=1pt}
]
\node[phase, draw=green!45!black, fill=green!15] (P1) at (0,2.05)
    {\textbf{Phase 1}\\No measures};
\node[phase, draw=orange!85!black, fill=orange!25] (P2) at (-2.85,0)
    {\textbf{Phase 2}\\Mild measures};
\node[phase, draw=red!70!black, fill=red!15] (P3) at (2.85,0)
    {\textbf{Phase 3}\\Strict measures};

\draw[switch, bend right=14] (P1.south west) to node[pos=0.33, above left=-1pt, lab] {$q(1,2)$} (P2.north east);
\draw[switch, bend right=14] (P2.north east) to node[pos=0.72, below right=-1pt, lab] {$q(2,1)$} (P1.south west);
\draw[switch, bend left=14] (P1.south east) to node[pos=0.33, above right=-1pt, lab] {$q(1,3)$} (P3.north west);
\draw[switch, bend left=14] (P3.north west) to node[pos=0.72, below left=-1pt, lab] {$q(3,1)$} (P1.south east);
\draw[switch] (P2.east) -- node[above, lab] {$q(2,3)$} (P3.west);
\draw[switch] ($(P3.west)+(0,-0.34)$) -- node[below=3pt, lab] {$q(3,2)$} ($(P2.east)+(0,-0.34)$);
\node[
    draw=black!35,
    rounded corners=5pt,
    inner sep=7pt,
    fit=(P1)(P2)(P3),
    label={[font=\small]above:{Markovian regime process}}
] {};
\end{tikzpicture}
\caption{Regime-switching mechanism.}
\end{subfigure}

\caption{Regime-switching SIR model with three intervention phases. In the general formulation, the phase-transition intensities may depend on the current epidemic state; the displayed labels $q(p,p')$ correspond to the state-independent special case.}
\label{fig:scheme}
\end{figure}

For $1\leq p\leq P$, set $l_a^{(p)}=\{(s,0,p):0\leq s\leq N-i_0\}$ and define
\begin{equation*}
S_a=\bigcup_{p=1}^{P}l_a^{(p)},\qquad S_T=S\setminus S_a.
\end{equation*}
The set $S_a$ is closed for the joint process, although its individual states need not be absorbing, since phase changes and direct immunity-acquisition events may continue after $I(t)$ has reached zero. The extinction time is
\begin{equation*}
T=\inf\{t\geq0:X(t)\in S_a\}.
\end{equation*}
The epidemic process is stopped at $T$.

Throughout Section~\ref{sec2}, all transition intensities are finite and
\begin{equation}
\label{eq:positive-recovery}
\gamma_*:=\min_{1\leq p\leq P}\gamma_p>0.
\end{equation}
Let $\mathbb{P}_x$ and $\mathbb{E}_x$ denote probability and expectation for the process started from $x\in S$. If $x=(s,i,p)$, conditioning on $(s,i,p)$ means conditioning on $X(0)=x$.

\begin{lemma}
\label{lem:absorption}
Under~\eqref{eq:positive-recovery}, $T<\infty$ almost surely for every initial state in $S_T$. Moreover, there exists $\eta_T>0$ such that
\begin{equation*}
\sup_{(s,i,p)\in S_T}\mathbb{E}\left[\exp\{\eta T\}\,\middle|\,(s,i,p)\right]<\infty,\qquad 0\leq\eta<\eta_T.
\end{equation*}
In particular, every positive integer moment of $T$ is finite.
\end{lemma}

The proof is given in Appendix~\ref{app:proofs}. Choose and fix
\begin{equation*}
0<\eta_0<\min\{\eta_T,\gamma_*\}.
\end{equation*}
Let $\mathbf{I}_m$ denote the identity matrix of order $m$, $\mathbf{0}_{m\times n}$ the $m\times n$ zero matrix, $\mathbf{1}_m$ the column vector of $m$ ones and $\mathbf{e}_r(m)$ the $r$th canonical column vector of $\mathbb{R}^m$. If the subscript is an event, $\mathbf{1}_A$ denotes its indicator. Let $\otimes$ denote the Kronecker product and let $\mathsf{T}$ denote transposition.

The states in $S_a$ are ordered first by increasing $s$ and then by increasing $p$. The states in $S_T$ are ordered first by increasing $s$, then by increasing $i$ and finally by increasing $p$. Under this ordering,
\begin{equation*}
\mathbf{Q}=\begin{pmatrix}\mathbf{Q}_{a,a}&\mathbf{0}\\[1ex]\mathbf{Q}_{T,a}&\mathbf{Q}_{T,T}\end{pmatrix}.
\end{equation*}
The upper-right block is zero because $S_a$ is closed; $\mathbf{Q}_{a,a}$ is not needed for the first entrance time of $S_a$.

For $0\leq s\leq N-i_0$, put $n_s=N-s$, $M_a=N-i_0+1$ and $|S_a|=P M_a$. Define
\begin{equation*}
\mathbf{B}=\operatorname{diag}(b_1,\ldots,b_P),\qquad \boldsymbol{\Gamma}=\operatorname{diag}(\gamma_1,\ldots,\gamma_P),\qquad \boldsymbol{\Psi}=\operatorname{diag}(\psi_1,\ldots,\psi_P).
\end{equation*}
For $1\leq i,j\leq n_s$, let the $(i,j)$ block of $\mathbf{Q}_{s,s}\in\mathbb{R}^{P n_s\times P n_s}$ be
\begin{equation*}
(\mathbf{Q}_{s,s})_{ij}=
\begin{cases}
\mathbf{Q}_P(s,i)-\dfrac{s\,i}{N}\,\mathbf{B}-i\,\boldsymbol{\Gamma}-s\,\boldsymbol{\Psi},&j=i,\\[1ex]
i\,\boldsymbol{\Gamma},&j=i-1,\quad 2\leq i\leq n_s,\\[1ex]
\mathbf{0}_{P\times P},&\text{otherwise}.
\end{cases}
\end{equation*}
For $1\leq s\leq N-i_0$, define $\mathbf{Q}_{s,s-1}\in\mathbb{R}^{P n_s\times P(n_s+1)}$ by
\begin{equation*}
(\mathbf{Q}_{s,s-1})_{ij}=\begin{cases}s\,\boldsymbol{\Psi},&j=i,\\[1ex]\dfrac{s\,i}{N}\,\mathbf{B},&j=i+1,\\[1ex]\mathbf{0}_{P\times P},&\text{otherwise},\end{cases}
\end{equation*}
where $1\leq i\leq n_s$ and $1\leq j\leq n_s+1$. Finally, define $\mathbf{Q}_{s,a}\in\mathbb{R}^{P n_s\times P M_a}$ by
\begin{equation*}
\mathbf{Q}_{s,a}=\bigl[\mathbf{e}_1(n_s)\mathbf{e}_{s+1}^{\mathsf{T}}(M_a)\bigr]\otimes\boldsymbol{\Gamma},\qquad 0\leq s\leq N-i_0.
\end{equation*}
Thus, the process enters  $S_a$ after a recovery, only when $i=1$. Hence,
\begin{equation*}
\mathbf{Q}_{T,T}=\begin{pmatrix}\mathbf{Q}_{0,0}&\mathbf{0}&\cdots&\mathbf{0}\\\mathbf{Q}_{1,0}&\mathbf{Q}_{1,1}&\ddots&\vdots\\\mathbf{0}&\ddots&\ddots&\mathbf{0}\\\vdots&\ddots&\mathbf{Q}_{N-i_0,N-i_0-1}&\mathbf{Q}_{N-i_0,N-i_0}\end{pmatrix},\qquad
\mathbf{Q}_{T,a}=\begin{pmatrix}\mathbf{Q}_{0,a}\\\mathbf{Q}_{1,a}\\\vdots\\\mathbf{Q}_{N-i_0,a}\end{pmatrix},
\end{equation*}
with the evident reduction when $N-i_0=0$. The block-bidiagonal form follows directly from~\eqref{eq:generator}.  Recoveries preserve $s$, whereas infections and direct immunity acquisition decrease $s$ by one.

\subsection{Joint transform of extinction time and infection count}
\label{subsec:joint-transform}

Let $\mathcal{J}$ be the set of jump times of $\mathcal{X}$ and define
\begin{equation*}
N^I(t)=\sum_{\tau\in\mathcal{J}:\,0<\tau\leq t}\mathbf{1}_{\{S(\tau)=S(\tau^-)-1,\;I(\tau)=I(\tau^-)+1\}},\qquad N^I=N^I(T).
\end{equation*}
The variable $N^I$ counts infections after the starting time of the calculation. Since each infection decreases the susceptible population by one, conditional on $X(0)=(s,i,p)\in S_T$, $N^I\in\{0,1,\ldots,s\}$. The total number of individuals infected from the starting time, including those infectious at that time, is
\begin{equation}
\label{eq:total-infected-definition}
C^I=I(0)+N^I.
\end{equation}
Under the initial condition $R(0)=0$, this is the total outbreak size. Conditional on $(s,i,p)$, $C^I=i+N^I$.

For $\operatorname{Re}(z)\geq0$ and $|u|\leq1$, define
\begin{equation}
\label{eq:joint-transform}
\Phi_{s,i,p}(z,u)=\mathbb{E}\left[\exp\{-zT\}u^{N^I}\,\middle|\,S(0)=s,I(0)=i,J(0)=p\right],
\end{equation}
with the convention $u^0=1$. On this domain, $\Phi_{s,i,p}$ is the joint Laplace--Stieltjes transform--probability generating function of $(T,N^I)$. Since $N^I\leq s$ and Lemma~\ref{lem:absorption} gives exponential moments of $T$, it extends analytically in $z$ to $\operatorname{Re}(z)>-\eta_0$ and polynomially in $u$ to every $u\in\mathbb{C}$. We use the same notation for this extension. For $i=0$,
\begin{equation*}
\Phi_{s,0,p}(z,u)=1,
\qquad 0\leq s\leq N-i_0,
\quad 1\leq p\leq P.
\end{equation*}
For $0\leq s\leq N-i_0$ and $1\leq i\leq N-s$, write
\begin{equation*}
\boldsymbol{\Phi}_{s,i}(z,u)=(\Phi_{s,i,1}(z,u),\ldots,\Phi_{s,i,P}(z,u))^{\mathsf{T}}.
\end{equation*}
Define
\begin{equation*}
\mathbf{A}_{s,i}(z)=z\mathbf{I}_P+\dfrac{s\,i}{N}\,\mathbf{B}+i\,\boldsymbol{\Gamma}+s\,\boldsymbol{\Psi}-\mathbf{Q}_P(s,i).
\end{equation*}

For transient states $x=(s,i,p)$ and $y=(s',i',p')$, let $\kappa(x,y)=\mathbf{1}_{\{(s',i',p')=(s-1,i+1,p)\}}$. Let $\mathbf{Q}_{T,T}^{[I]}$ have entries $q_{x,y}\kappa(x,y)$, and put
\begin{equation*}
\mathbf{Q}_{T,T}(u)=\mathbf{Q}_{T,T}+(u-1)\mathbf{Q}_{T,T}^{[I]}.
\end{equation*}
Thus, infection transitions are multiplied by $u$ and all other entries are unchanged.

\begin{theorem}
\label{th:joint-transform}
For $\operatorname{Re}(z)>-\eta_0$ and $u\in\mathbb{C}$, the column vector $\boldsymbol{\Phi}(z,u)$ of the values in~\eqref{eq:joint-transform}, ordered as $S_T$, is the unique solution of
\begin{equation}
\label{eq:global-joint-system}
\left(z\mathbf{I}_{|S_T|}-\mathbf{Q}_{T,T}(u)\right)\boldsymbol{\Phi}(z,u)=\mathbf{Q}_{T,a}\mathbf{1}_{|S_a|}.
\end{equation}
Equivalently, $\boldsymbol{\Phi}_{s,0}(z,u)=\mathbf{1}_P$ and
\begin{equation}
\label{eq:joint-recursion-zero}
\mathbf{A}_{0,i}(z)\boldsymbol{\Phi}_{0,i}(z,u)=i\,\boldsymbol{\Gamma}\boldsymbol{\Phi}_{0,i-1}(z,u),\qquad 1\leq i\leq N,
\end{equation}
whereas, for $1\leq s\leq N-i_0$ and $1\leq i\leq N-s$,
\begin{align}
\mathbf{A}_{s,i}(z)\boldsymbol{\Phi}_{s,i}(z,u)
={}&i\,\boldsymbol{\Gamma}\boldsymbol{\Phi}_{s,i-1}(z,u)+s\,\boldsymbol{\Psi}\boldsymbol{\Phi}_{s-1,i}(z,u)
\notag\\
&+u\,\dfrac{s\,i}{N}\,\mathbf{B}\boldsymbol{\Phi}_{s-1,i+1}(z,u).
\label{eq:joint-recursion}
\end{align}
The vectors are computed by increasing $s$ and, within each level, by increasing $i$.
\end{theorem}

The proof is given in Appendix~\ref{app:proofs}. The matrix $\mathbf{A}_{s,i}(z)$ is nonsingular on the stated domain. If $\lambda_{p,p'}(s,i)=0$ for all $p'\neq p$ and all $(s,i)$, then the phase equations separate. If $\lambda_{p,p'}(s,i)=q(p,p')$, the exogenous switching model is recovered. In the general case, the $P$ phase values for a fixed pair $(s,i)$ are solved together.

\subsection{Consequences of the joint transform}
\label{subsec:consequences}

The extinction-time transform and the infection-count probability generating function are
\begin{equation*}
y_{s,i,p}(z)=\Phi_{s,i,p}(z,1),\qquad g_{s,i,p}(u)=\Phi_{s,i,p}(0,u).
\end{equation*}
The probability generating function of $C^I$ conditional on $(s,i,p)$ is $u^i g_{s,i,p}(u)$. If $\boldsymbol{y}_{s,i}(z)=(y_{s,i,1}(z),\ldots,y_{s,i,P}(z))^{\mathsf{T}}$, then Theorem~\ref{th:joint-transform} gives
\begin{equation*}
\mathbf{A}_{0,i}(z)\boldsymbol{y}_{0,i}(z)=i\,\boldsymbol{\Gamma}\boldsymbol{y}_{0,i-1}(z),\qquad 1\leq i\leq N,
\end{equation*}
and
\begin{align*}
\mathbf{A}_{s,i}(z)\boldsymbol{y}_{s,i}(z)
={}&i\,\boldsymbol{\Gamma}\boldsymbol{y}_{s,i-1}(z)+s\,\boldsymbol{\Psi}\boldsymbol{y}_{s-1,i}(z)
\notag\\
&+\dfrac{s\,i}{N}\,\mathbf{B}\boldsymbol{y}_{s-1,i+1}(z),
\end{align*}
for $1\leq s\leq N-i_0$ and $1\leq i\leq N-s$, with $\boldsymbol{y}_{s,0}(z)=\mathbf{1}_P$.

For $0\leq n\leq s$ and $i\geq1$, define
\begin{equation}
\label{eq:coefficient-transform}
h_{s,i,p}^{(n)}(z)=\mathbb{E}\left[\exp\{-zT\}\mathbf{1}_{\{N^I=n\}}\,\middle|\,(s,i,p)\right],
\end{equation}
and set $h_{s,i,p}^{(n)}(z)=0$ outside this range. With $\mathbf{h}_{s,i}^{(n)}(z)=(h_{s,i,1}^{(n)}(z),\ldots,h_{s,i,P}^{(n)}(z))^{\mathsf{T}}$,
\begin{equation}
\label{eq:joint-polynomial-expansion}
\boldsymbol{\Phi}_{s,i}(z,u)=\sum_{n=0}^{s}u^n\mathbf{h}_{s,i}^{(n)}(z).
\end{equation}

\begin{corollary}
\label{cor:coefficient-recursion}
The boundary values are $\mathbf{h}_{s,0}^{(0)}(z)=\mathbf{1}_P$ and $\mathbf{h}_{s,0}^{(n)}(z)=\mathbf{0}_{P\times1}$ for $n\neq0$. For $1\leq i\leq N$,
\begin{equation}
\label{eq:coefficient-zero}
\mathbf{A}_{0,i}(z)\mathbf{h}_{0,i}^{(0)}(z)=i\,\boldsymbol{\Gamma}\mathbf{h}_{0,i-1}^{(0)}(z),\qquad \mathbf{h}_{0,i}^{(n)}(z)=\mathbf{0}_{P\times1}\;(n\neq0).
\end{equation}
For $1\leq s\leq N-i_0$, $1\leq i\leq N-s$ and $0\leq n\leq s$,
\begin{align}
\mathbf{A}_{s,i}(z)\mathbf{h}_{s,i}^{(n)}(z)
={}&i\,\boldsymbol{\Gamma}\mathbf{h}_{s,i-1}^{(n)}(z)+s\,\boldsymbol{\Psi}\mathbf{h}_{s-1,i}^{(n)}(z)
\notag\\
&+\dfrac{s\,i}{N}\,\mathbf{B}\mathbf{h}_{s-1,i+1}^{(n-1)}(z).
\label{eq:coefficient-recursion}
\end{align}
\end{corollary}

The proof is given in Appendix~\ref{app:proofs}. For $i\geq1$, setting $z=0$ in~\eqref{eq:coefficient-transform} gives the probability mass function of $N^I$:
\begin{equation*}
x_{s,i,p}(n)=\mathbb{P}(N^I=n\,\mid\,(s,i,p))=h_{s,i,p}^{(n)}(0),\qquad 0\leq n\leq s.
\end{equation*}
For $i=0$, set $x_{s,0,p}(0)=1$ and $x_{s,0,p}(n)=0$ for $n\neq0$.
Also,
\begin{equation}
\label{eq:total-infected-pmf}
\mathbb{P}(C^I=c\,\mid\,(s,i,p))=h_{s,i,p}^{(c-i)}(0),\qquad i\leq c\leq i+s.
\end{equation}

If $h_{s,i,p}^{(n)}(0)>0$, then the conditional Laplace--Stieltjes transform of $T$, given $N^I=n$, is
\begin{equation}
\label{eq:conditional-lst}
\mathbb{E}\left[\exp\{-zT\}\,\middle|\,N^I=n,(s,i,p)\right]=\dfrac{h_{s,i,p}^{(n)}(z)}{h_{s,i,p}^{(n)}(0)},\qquad \operatorname{Re}(z)>-\eta_0.
\end{equation}
For each integer $k\geq1$,
\begin{equation}
\label{eq:conditional-moments}
\mathbb{E}\left[T^k\,\middle|\,N^I=n,(s,i,p)\right]=\dfrac{(-1)^k\left.\dfrac{\partial^k}{\partial z^k}h_{s,i,p}^{(n)}(z)\right|_{z=0}}{h_{s,i,p}^{(n)}(0)}.
\end{equation}
The proof of these two identities is included in Appendix~\ref{app:proofs}.

For nonnegative integers $m$ and $r$, define $(m)_0=1$ and $(m)_r=m!/(m-r)!$ for $1\leq r\leq m$, while $(m)_r=0$ for $r>m$. For nonnegative integers $k$ and $r$, set
\begin{equation*}
\mu_{s,i,p}^{(k,r)}=\mathbb{E}\left[T^k(N^I)_r\,\middle|\,(s,i,p)\right], \quad k\geq0, r\geq0,
\end{equation*}
and let $\boldsymbol{\mu}_{s,i}^{(k,r)}=(\mu_{s,i,1}^{(k,r)},\ldots,\mu_{s,i,P}^{(k,r)})^{\mathsf{T}}$. These moments are finite because $N^I\leq s$ and all positive integer moments of $T$ are finite.

We adopt the convention that
$\boldsymbol{\mu}_{s,i}^{(k,r)}=\mathbf{0}_{P\times1}$ whenever $k<0$ or $r<0$.
Thus, the terms with indices $(k-1,r)$ and $(k,r-1)$ in the recursions below vanish automatically when $k=0$ and $r=0$, respectively.

\begin{theorem}
\label{th:mixed-moments}
The boundary values are $\boldsymbol{\mu}_{s,0}^{(0,0)}=\mathbf{1}_P$ and $\boldsymbol{\mu}_{s,0}^{(k,r)}=\mathbf{0}_{P\times1}$ for $k+r\geq1$. Also, $\boldsymbol{\mu}_{s,i}^{(0,0)}=\mathbf{1}_P$ for every transient state. For $k+r\geq1$,
\begin{equation}
\label{eq:mixed-recursion-zero}
\mathbf{A}_{0,i}(0)\boldsymbol{\mu}_{0,i}^{(k,r)}=k\boldsymbol{\mu}_{0,i}^{(k-1,r)}+i\,\boldsymbol{\Gamma}\boldsymbol{\mu}_{0,i-1}^{(k,r)},\qquad 1\leq i\leq N,
\end{equation}
and, for $1\leq s\leq N-i_0$ and $1\leq i\leq N-s$,
\begin{equation}
\label{eq:mixed-recursion}
\mathbf{A}_{s,i}(0)\boldsymbol{\mu}_{s,i}^{(k,r)}=k\boldsymbol{\mu}_{s,i}^{(k-1,r)}+i\,\boldsymbol{\Gamma}\boldsymbol{\mu}_{s,i-1}^{(k,r)}+s\,\boldsymbol{\Psi}\boldsymbol{\mu}_{s-1,i}^{(k,r)}+\dfrac{s\,i}{N}\,\mathbf{B}\left(\boldsymbol{\mu}_{s-1,i+1}^{(k,r)}+r\boldsymbol{\mu}_{s-1,i+1}^{(k,r-1)}\right).
\end{equation}
\end{theorem}

The proof is given in Appendix~\ref{app:proofs}. The extinction-time moments are obtained by setting $r=0$ in the formulas of Theorem~\ref{th:mixed-moments}. In particular, if $\boldsymbol{m}_{s,i}^{(k)}=\boldsymbol{\mu}_{s,i}^{(k,0)}$, then
\begin{align}
\mathbf{A}_{0,i}(0)\boldsymbol{m}_{0,i}^{(k)}&=k\boldsymbol{m}_{0,i}^{(k-1)}+i\,\boldsymbol{\Gamma}\boldsymbol{m}_{0,i-1}^{(k)},\qquad 1\leq i\leq N,
\label{eq:extinction-moment-zero}\\
\mathbf{A}_{s,i}(0)\boldsymbol{m}_{s,i}^{(k)}&=k\boldsymbol{m}_{s,i}^{(k-1)}+i\,\boldsymbol{\Gamma}\boldsymbol{m}_{s,i-1}^{(k)}+s\,\boldsymbol{\Psi}\boldsymbol{m}_{s-1,i}^{(k)}+\dfrac{s\,i}{N}\,\mathbf{B}\boldsymbol{m}_{s-1,i+1}^{(k)}.
\label{eq:extinction-moment-recursion}
\end{align}
The factorial moments of $N^I$ are obtained by setting $k=0$. In the following identities, $\mathbb{V}$ denotes variance:
\begin{align*}
\mathbb{E}[T\,\mid\,(s,i,p)]&=\mu_{s,i,p}^{(1,0)},
&\mathbb{V}(T\,\mid\,(s,i,p))&=\mu_{s,i,p}^{(2,0)}-(\mu_{s,i,p}^{(1,0)})^2,\\
\mathbb{E}[N^I\,\mid\,(s,i,p)]&=\mu_{s,i,p}^{(0,1)},
&\mathbb{V}(N^I\,\mid\,(s,i,p))&=\mu_{s,i,p}^{(0,2)}+\mu_{s,i,p}^{(0,1)}-(\mu_{s,i,p}^{(0,1)})^2,\\
\operatorname{Cov}(T,N^I\,\mid\,(s,i,p))&=\mu_{s,i,p}^{(1,1)}-\mu_{s,i,p}^{(1,0)}\mu_{s,i,p}^{(0,1)}.
\end{align*}
Since $C^I=i+N^I$ under conditioning on $(s,i,p)$,
\begin{equation*}
\mathbb{V}(C^I\,\mid\,(s,i,p))=\mathbb{V}(N^I\,\mid\,(s,i,p)),\qquad
\operatorname{Cov}(T,C^I\,\mid\,(s,i,p))=\operatorname{Cov}(T,N^I\,\mid\,(s,i,p)).
\end{equation*}

\subsection{Computational aspects}
\label{subsec:computational}

For a fixed pair $(z,u)$, the recursion in Theorem~\ref{th:joint-transform} requires one $P\times P$ linear system for each admissible pair $(s,i)$ with $i\geq1$. Since
\begin{equation*}
|S_T|=P\sum_{s=0}^{N-i_0}(N-s)=P\dfrac{(N-i_0+1)(N+i_0)}{2},
\end{equation*}
the dense cost is $O(P^3N^2)$ for fixed $i_0$. This order is unchanged by the state dependence of $\lambda_{p,p'}(s,i)$, because only the entries of the $P\times P$ matrices change from one epidemic state to another. A dense solution of the full system~\eqref{eq:global-joint-system} would have order $O(|S_T|^3)$, although sparse solvers can reduce this cost. For the complete set of coefficient transforms $h_{s,i,p}^{(n)}(z)$, the matrices $\mathbf{A}_{s,i}(z)$ are factorized once for each $(s,i)$ and reused over $n$. Thus, the cost is of order $O(P^3N^2+P^2N^3)$ for fixed $i_0$. The identities
\begin{equation*}
\Phi_{s,i,p}(0,1)=1,
\qquad
\sum_{n=0}^{s}h_{s,i,p}^{(n)}(0)=1,
\qquad
\sum_{c=i}^{i+s}\mathbb{P}(C^I=c\,\mid\,(s,i,p))=1
\end{equation*}
provide basic checks on the implementation.

\section{Application to mpox incidence data}
\label{sec3}

This section provides a model-based illustration using weekly incidence data from the 2022--2023 mpox outbreak in Luxembourg. The empirical analysis is intentionally parsimonious: a baseline one-phase SIR model is first calibrated to the incidence series, and alternative regime-switching mechanisms are then evaluated conditionally on the fitted baseline parameters. Thus, only the baseline epidemiological parameters are informed by the observed incidence data; the switching intensities and intervention-specific effects are specified rather than estimated from the historical intervention process. The scenarios should consequently be interpreted as conditional comparisons of model mechanisms, not as estimates of the causal or historical effects of interventions implemented in Luxembourg.

The recursive probability mass functions, transforms, and moments are evaluated exactly up to numerical solution of the corresponding linear systems. The continuous extinction-time densities are subsequently obtained by numerical inversion of the Laplace--Stieltjes transform using the Abate--Whitt method~\citep{Abate1995}. For each initial state and each intervention specification,
we compute the probability mass function of the total number of infected
individuals and the distributional characteristics of the extinction time. The total number of infected individuals is defined in Eq.~\eqref{eq:total-infected-definition}, where $N^I$ denotes the number of infections generated after the initial state.
Thus, $C^I$ includes the individuals infectious at time zero. The probability
mass function of $C^I$ is obtained from~\eqref{eq:total-infected-pmf}, whereas
the first two moments of $T$ are obtained from
\eqref{eq:extinction-moment-zero}--\eqref{eq:extinction-moment-recursion}. The density curves of the extinction time are obtained by applying the Abate--Whitt numerical inversion method to the Laplace--Stieltjes transform \(\Phi_{s_0,i_0,p_0}(z,1)\) for the selected initial state \((s_0,i_0,p_0)\in S_T\). 
In the implementation, the density was evaluated on a finite time grid for each scenario group. 
The upper endpoint of the grid was set to
\[
\max_j \left\{ \mathbb{E}_j[T] + 4\,\operatorname{SD}_j(T) \right\},
\]
where the maximum is taken over the scenarios plotted in the same figure, and the grid contained 120 points. 
For each grid value \(t>0\), the Abate--Whitt Euler inversion was applied to the transform
\[
z \mapsto \Phi_{s_0,i_0,p_0}(z,1)
\]
using the standard choice of 
\[
z_k=\frac{A+2\pi \mathrm{i}k}{2t},
\qquad A=-log(10^{-8})\approx 18.4,
\]
with 15 initial terms and 11 Euler averaging terms \citep{Abate1995}. 
The approximation was computed from the real parts of the transform evaluations, with alternating signs and binomial Euler weights. 
Numerical values that were non-finite or negative only because of inversion error were set to zero before plotting.

We examine four groups of intervention scenarios. The first compares one-phase control settings with different levels of transmission reduction. The second compares earlier and delayed transitions from no measures to strict control. The third evaluates the additional effect of vaccination-supported immunity when strict control is already in place. These three groups correspond to the state-independent
special case
\begin{equation*}
\lambda_{r,\ell}(s,i)=q(r,\ell),
\qquad r\neq \ell,
\end{equation*}
for all admissible $(s,i)$. The fourth group uses an explicitly state-dependent transition intensity, in order to show how the general formulation can represent policy escalation that reacts to the current number of infectious individuals.

\subsection{Data and baseline calibration}
\label{sec3.1}

\paragraph{Data and temporal aggregation.}
The empirical series consists of publicly available daily mpox case counts for Luxembourg obtained from the Our World in Data mpox repository~\citep{EdouardMathieu2022}; the data can be accessed and downloaded from \url{https://ourworldindata.org/mpox}. The analysis used 
reported daily mpox case counts over the period from 16 June 2022 to 13 September 2022. The one-phase model was fitted to the complete period. Accordingly, the fitted transmission coefficient is interpreted below as an effective baseline transmission coefficient rather than as an estimate of transmission in a strictly uncontrolled phase.

The daily counts were aggregated into calendar weeks using Monday as the week start. Since the observation window starts on 16/06/2022 and ends on 13/09/2022, the first and last aggregated weeks are partial. These weeks were retained in the calibration, giving \(K=14\) weekly counts, denoted by \(y_1,\ldots,y_K\). Weekly aggregation reduces sparsity and day-of-reporting fluctuations in the daily series and yields observations consistent with the weekly time unit adopted for model calibration.

\paragraph{Deterministic calibration model and observation models.}
A one-phase deterministic SIR model without regime switching was fitted to the weekly incidence series. Time is measured in weeks, \(R(0)=0\), and the population size in the fitted effective system is \(N=S(0)+I(0)+R(0)\). As a modeling approximation, the mean infectious period was fixed at three weeks, consistently with the reported two-to-four-week duration of the symptomatic phase of mpox~\citep{WHOmpox}. Accordingly, \(\gamma=1/3\) per week, where \(\gamma\) is the one-phase counterpart of the phase-specific recovery intensity \(\gamma_p\) defined in Section~\ref{subsec:model}. This choice should be interpreted as a pragmatic calibration assumption rather than as a direct estimate of the infectious period from the Luxembourg data. For
\begin{equation*}
\theta=\bigl(b,S(0),I(0)\bigr),
\end{equation*}
the calibration system is
\begin{equation*}
\frac{\mathrm dS(t)}{\mathrm dt}=-\frac{b}{N}S(t)I(t),
\qquad
\frac{\mathrm dI(t)}{\mathrm dt}=\frac{b}{N}S(t)I(t)-\gamma I(t),
\qquad
\frac{\mathrm dR(t)}{\mathrm dt}=\gamma I(t),
\end{equation*}
with the cumulative number of infections generated after time zero defined by
\begin{equation*}
\frac{\mathrm dC(t)}{\mathrm dt}=\frac{b}{N}S(t)I(t),
\qquad C(0)=0.
\end{equation*}
The model-implied incidence during week \(t\) is therefore
\begin{equation*}
\mu_t(\theta)=C(t)-C(t-1),
\qquad 1\leq t\leq K.
\end{equation*}

Under the Poisson observation model, the weekly counts are conditionally independent and
\begin{equation*}
y_t\mid\theta\sim\operatorname{Poisson}\!\left(\mu_t(\theta)\right),
\qquad t=1,\ldots,K.
\end{equation*}
The corresponding log-likelihood is
\begin{equation*}
\ell_{\mathrm P}(\theta)
=\sum_{t=1}^{K}\left[
 y_t\log\{\mu_t(\theta)\}-\mu_t(\theta)-\log(y_t!)
\right].
\end{equation*}

To assess possible extra-Poisson variability, a negative-binomial
observation model with the same conditional mean $\mu_t(\theta)$
and size (inverse-dispersion) parameter $\kappa$ was also fitted.
The implemented mean--size parameterization was
\begin{equation*}
\operatorname{Var}(y_t\mid\theta,\kappa)
=
\mu_t(\theta)+\frac{\mu_t(\theta)^2}{\kappa}.
\end{equation*}

Under this convention, the Poisson model is recovered as \(\kappa\to\infty\). The associated negative-binomial log-likelihood was maximized over \((\theta,\kappa)\). 
The admissible upper bound for the size parameter was
$\kappa_{\max}=20000$.

For both observation models, the negative log-likelihood was minimized with the L-BFGS-B algorithm~\citep{Byrd1995}. The box constraints were
\begin{equation*}
b\in[0.05,3.00],
\qquad
S(0)\in\left[
\max\left\{20,\sum_{t=1}^K y_t\right\},
\max\left\{100,5\sum_{t=1}^K y_t\right\}
\right],
\qquad
I(0)\in[10^{-4},10],
\end{equation*}
and, for the negative-binomial fit,
\begin{equation*}
\kappa\in[0.05,20000].
\end{equation*}
For the data used here, \(\sum_{t=1}^K y_t=55\), so that the bounds for \(S(0)\) were \([55,275]\).

The first Poisson starting value was obtained from a simple initialization rule based on the early growth of the weekly series, giving
\begin{equation*}
\theta_{\mathrm P}^{(0)}=(0.552,66,1).
\end{equation*}
For the negative-binomial fit, the first starting value was
\begin{equation*}
\theta_{\mathrm{NB}}^{(0)}=(0.8,66,1),
\qquad
\kappa^{(0)}=20.
\end{equation*}
Additional starting values were generated randomly within the admissible parameter bounds, with \(\kappa\) sampled on a logarithmic scale.

The ODE system was integrated using a fixed-step fourth-order Runge--Kutta scheme with 100 substeps per week, corresponding to a step size of 0.01 weeks. The likelihood maximization was performed on the original bounded parameter scale. Approximate Wald confidence intervals were then obtained from the inverse observed information matrix on an unconstrained logit-transformed bounded scale and mapped back to the original scale.

\paragraph{Choice of the Poisson observation model.}
In the negative-binomial fit,
the estimated size parameter reached its upper admissible bound
$\widehat\kappa=20000=\kappa_{\max}$,
which corresponds under the implemented parameterization to negligible extra-Poisson variation. As an additional diagnostic, the Pearson dispersion statistic for the Poisson fit was
\begin{equation*}
\widehat D_P
=\frac{1}{K-d}\sum_{t=1}^{K}
\frac{\{y_t-\mu_t(\widehat\theta)\}^2}{\mu_t(\widehat\theta)},
\end{equation*}
where \(d=3\) is the number of fitted parameters. The resulting value, \(\widehat D_P=0.79\), provides no evidence of residual overdispersion. The Poisson model was therefore retained as the more parsimonious observation model for the baseline calibration.

\paragraph{Estimated parameters and integer initial state.}
Following \citet{Bettencourt2009} and \citet{Kypraios2009}, the initial numbers of susceptible and infectious individuals were treated as effective unknown quantities and estimated jointly with the transmission coefficient. This is appropriate because the Luxembourg outbreak involved only a small part of the national population; \(S(0)\) is therefore interpreted as the size of an effective at-risk subpopulation. Under the retained Poisson model, the estimates were
\begin{equation*}
\widehat b=0.99,
\qquad
\widehat S(0)=61.88,
\qquad
\widehat I(0)=0.28,
\end{equation*}
with approximate \(95\%\) Wald confidence intervals
\begin{equation*}
[0.80,1.21],
\qquad
[55.54,120.87],
\qquad
[0.06,1.12],
\end{equation*}
respectively.

The fitted values \(\widehat S(0)\) and \(\widehat I(0)\) are continuous effective initial conditions, whereas the finite-state Markov chain requires integer-valued compartments. For the numerical scenarios, \(\widehat S(0)\) was rounded to the nearest integer, while \(\widehat I(0)\) was rounded subject to a lower bound of one infectious individual; thus,
\begin{equation*}
S(0)=\operatorname{round}\{\widehat S(0)\}=62,
\qquad
I(0)=\max\!\left\{1,\operatorname{round}\{\widehat I(0)\}\right\}=1,
\qquad
R(0)=0.
\end{equation*}
These integer values constitute an initialization convention for the finite-state Markov chain and should not be interpreted as separate parameter estimates. Hence, \(N=63\). The lower bound of one infectious individual avoids initializing the Markov chain in the absorbing disease-free state. All scenarios below use this initial epidemic state and \(\widehat b=0.99\) as the effective baseline transmission coefficient. In the switching scenarios, the process starts in the no-measures regime, \(J(0)=p_0=1\).

\paragraph{Baseline goodness of fit.}
Figure~\ref{fig:baseline-calibration-fit} compares the observed weekly incidence \(y_t\) with the fitted Poisson means \(\widehat\mu_t=\mu_t(\widehat\theta)\). The fitted trajectory provides a reasonable parsimonious representation of the overall epidemic curve, capturing the rise in incidence, the peak region, and the subsequent decline. The main discrepancies concern local week-to-week fluctuations and the exact height of individual weekly counts, which are smoothed by the deterministic one-phase SIR trajectory. Thus, the fitted model should be interpreted as an effective baseline calibration for the scenario analysis, rather than as a detailed reconstruction of the observed incidence path.

\begin{figure}[htbp]
    \centering
    \includegraphics[width=0.8\textwidth]{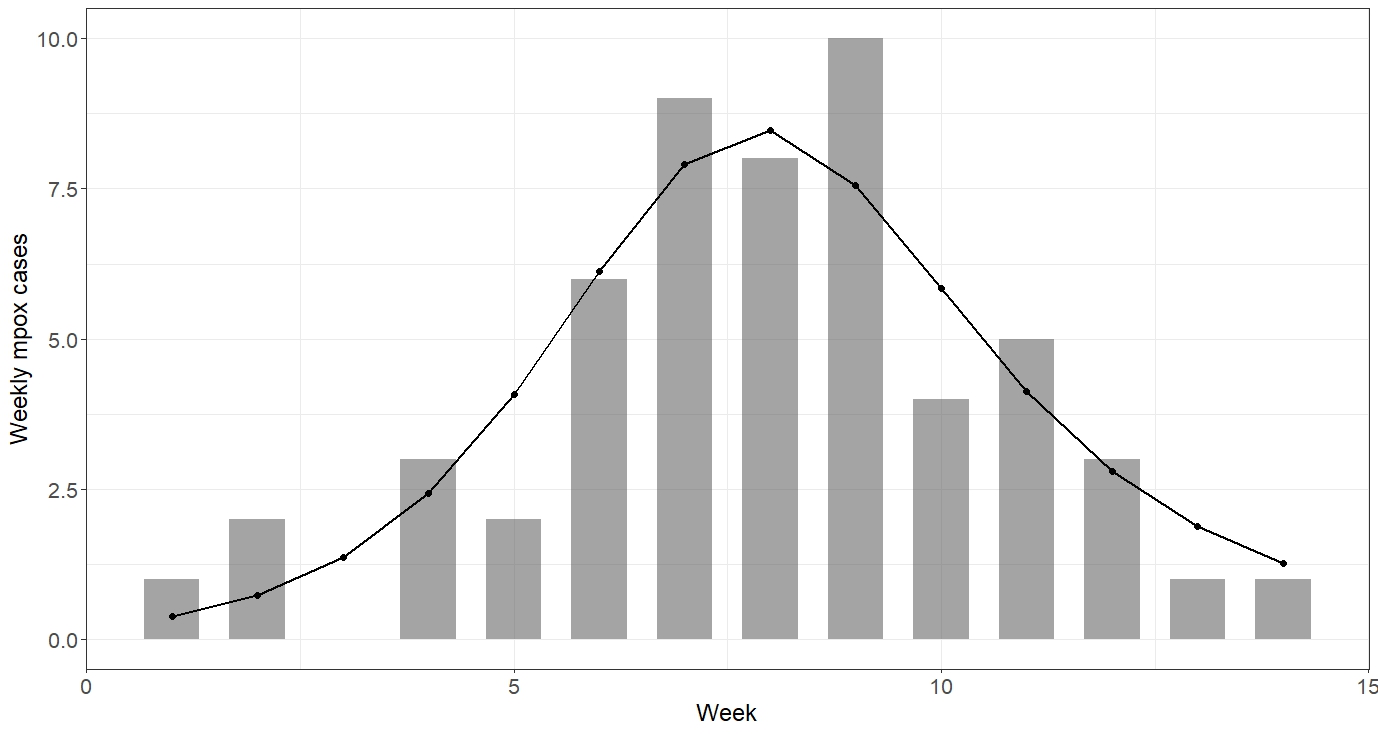}
    \caption{Observed weekly mpox incidence \(y_t\) and fitted Poisson means \(\widehat\mu_t\) under the one-phase baseline calibration. 
    Bars represent the observations, and the solid line represents the fitted means.}
    \label{fig:baseline-calibration-fit}
\end{figure}

\paragraph{Specification of scenario parameters.}
The intervention-specific values used in Scenarios~1--4 are illustrative and are chosen to separate the effects of intervention intensity, intervention timing, direct immunity acquisition, and state-dependent escalation. The mild and strict regimes are represented by \(25\%\) and \(50\%\) reductions of the fitted baseline transmission coefficient, respectively. The constant switching intensities are chosen to represent earlier and delayed escalation from no measures to strict control; for each constant switching intensity \(q\), the corresponding nominal mean waiting time is \(1/q\). Thus, \(q=0.35\) corresponds to approximately \(2.86\) weeks, whereas \(q=0.05\) corresponds to \(20\) weeks. The direct immunity-acquisition rate \(\psi=0.05\) is used as an illustrative vaccination-supported removal mechanism. The coefficients and threshold in the state-dependent escalation function~\eqref{eq:state-dependent-scenario-rate} are chosen so that the transition rate is low when the number of infectious individuals is small and increases smoothly as the infectious population grows. The resulting scenarios should therefore be interpreted as conditional model comparisons, not as estimates of the historical effects of interventions implemented in Luxembourg.

\subsection{\texorpdfstring{Scenario 1: effect of control intensity}{Scenario 1: effect of control intensity}}
\label{sec3.2}

The first scenario compares three one-phase control settings: no measures, mild measures, and strict measures. There is no vaccination-induced immunity, so that $\psi_p=0$ in all cases. There is also no phase switching. Equivalently, each curve is computed from a one-phase model.

The regimes differ only in their transmission coefficient:
\begin{equation*}
b_1=\widehat b=0.99,
\qquad
b_2=0.75\widehat b=0.743,
\qquad
b_3=0.50\widehat b=0.495.
\end{equation*}
This comparison separates transmission reduction from the timing of implementation and from vaccination.

Figure~\ref{fig:pmf_total_infections_scenario1} gives the distribution of $C^I$. Increasing control intensity moves probability mass towards smaller outbreak sizes. The mean of $C^I$ decreases from $38.52$ under no measures to $28.34$ under mild measures and to $13.07$ under strict measures. The corresponding standard deviations are $27.31$, $25.60$, and $17.08$. Thus, in this comparison, stronger control reduces both the mean infection count and its absolute dispersion.

Figure~\ref{fig:pdf_extinction_time_scenario1} gives the corresponding distributions of $T$. The mean extinction times are $15.02$, $14.16$, and $10.19$ weeks under no, mild, and strict measures, respectively. The standard deviations are $10.95$, $12.58$, and $11.86$. Thus, strict measures also reduce the expected time to extinction. The dispersion of $T$ is not monotone, however, and both controlled regimes retain a right tail. This indicates that infection count and extinction time should be examined together.

\begin{figure}[htbp]
    \centering
    \includegraphics[width=0.8\textwidth]{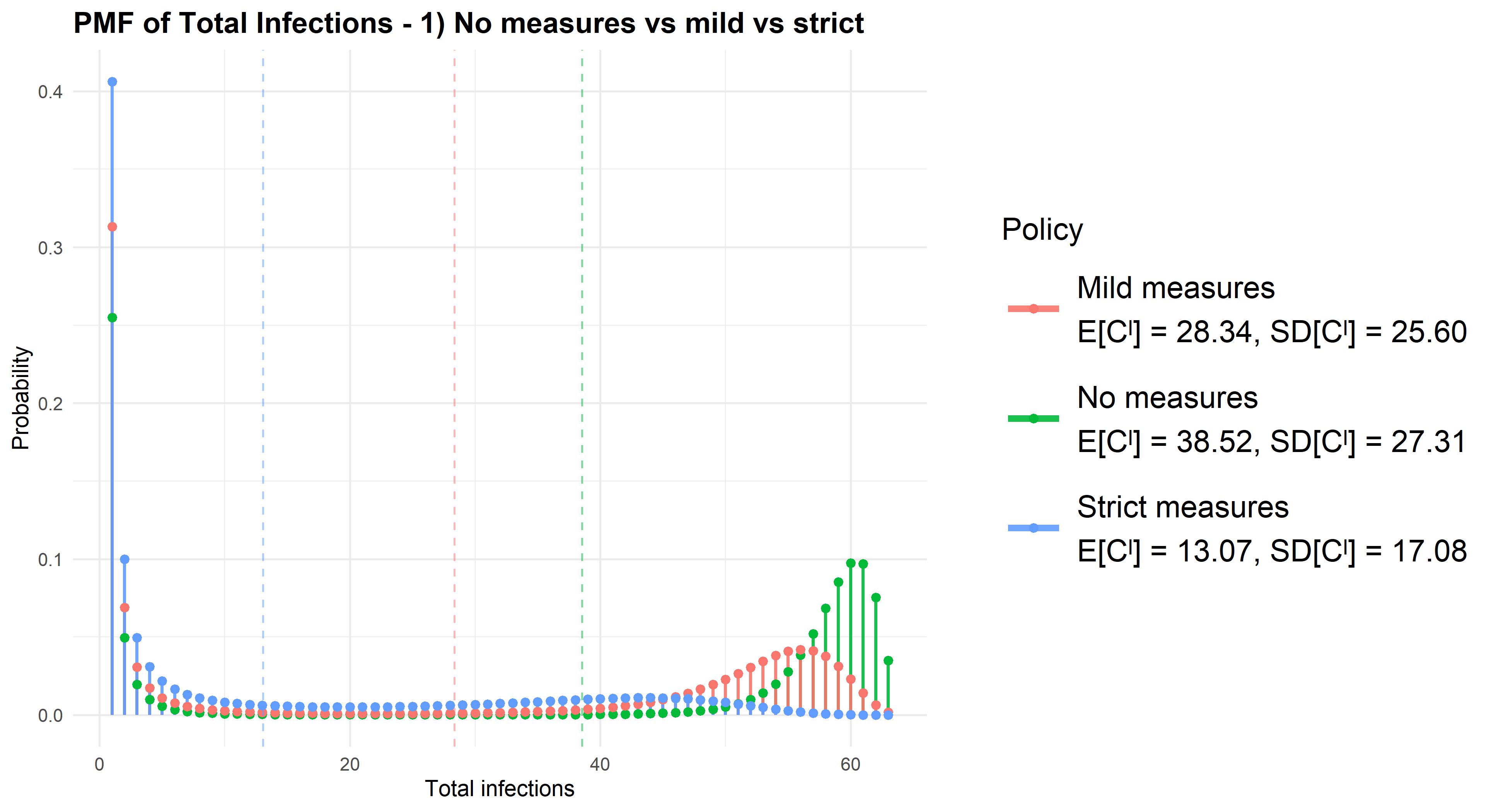}
    \caption{PMF of the total number of infected individuals under the one-phase control settings considered in Scenario~1. Dashed vertical lines indicate the corresponding means.}
    \label{fig:pmf_total_infections_scenario1}
\end{figure}

\begin{figure}[htbp]
    \centering
    \includegraphics[width=0.8\textwidth]{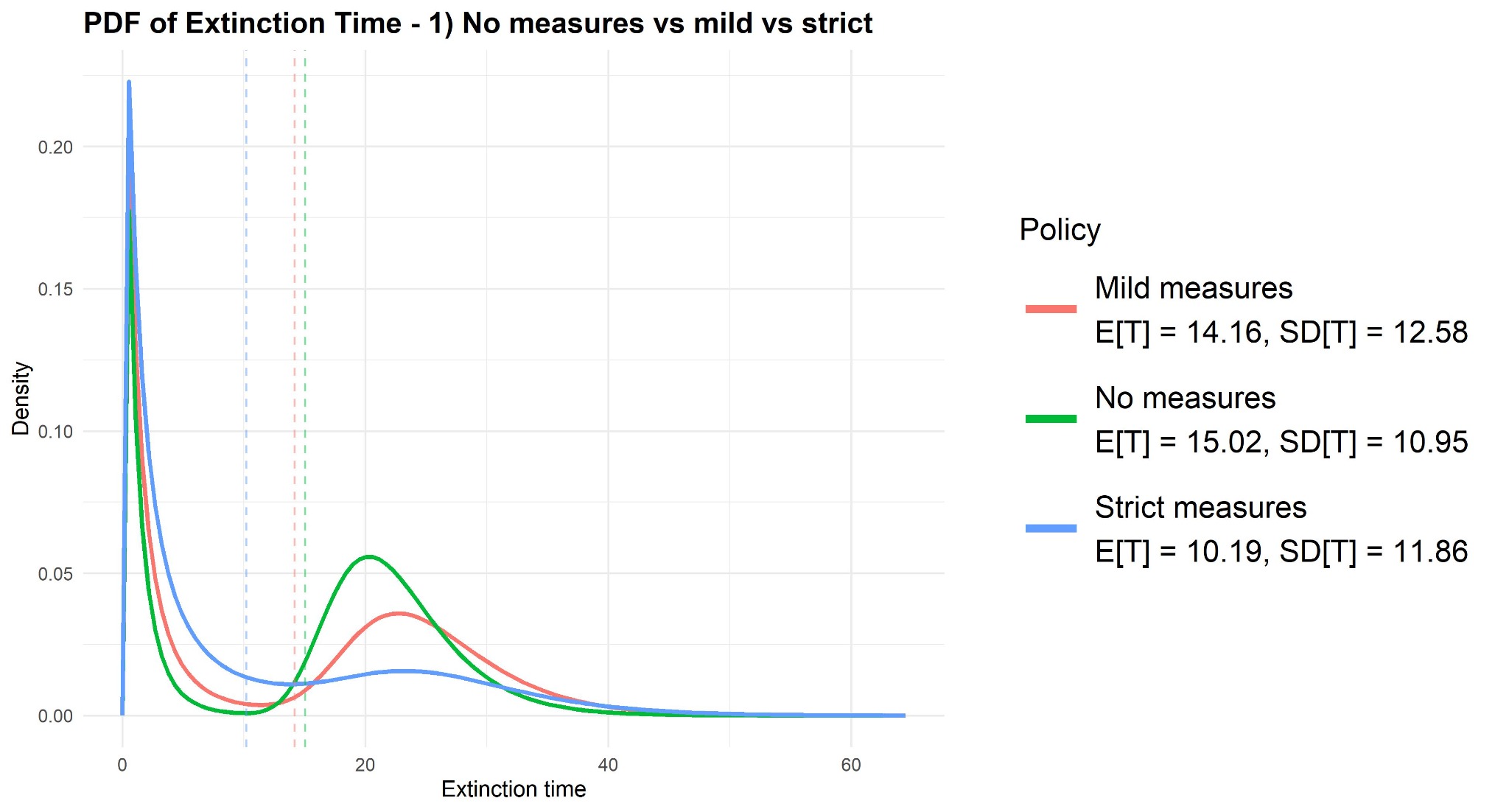}
    \caption{Density of the extinction time under the one-phase control settings considered in Scenario~1. Time is measured in weeks, and dashed vertical lines indicate the corresponding means.}
    \label{fig:pdf_extinction_time_scenario1}
\end{figure}

\FloatBarrier
\subsection{Scenario 2: effect of intervention timing}
\label{sec3.3}

The second scenario compares two timings for the introduction of strict measures. The process starts in the no-measures phase and may move to the strict phase. The reverse transition is excluded, so that $\lambda_{2,1}(s,i)=0$. The two switching rates are
\begin{equation*}
\lambda^{\mathrm{E}}_{1,2}(s,i)=0.35,
\qquad
\lambda^{\mathrm{D}}_{1,2}(s,i)=0.05,
\end{equation*}
where the superscripts $\mathrm{E}$ and $\mathrm{D}$ denote
early and delayed escalation, respectively.
These rates correspond to nominal mean waiting times of approximately $2.86$ and $20$ weeks for the phase-switching clock while the process remains in the initial regime. The transition rate is state-independent in both cases. The epidemic parameters are
\begin{equation*}
b_1=0.99,
\qquad
b_2=0.495,
\qquad
\gamma_1=\gamma_2=1/3,
\qquad
\psi_1=\psi_2=0.
\end{equation*}

The effect of intervention timing on infection burden is apparent in Figure~\ref{fig:pmf_total_infections_scenario2}: earlier strict intervention shifts the distribution of $C^I$ towards smaller values. Its mean decreases from $33.52$ under delayed intervention to $21.82$ under early intervention, while the corresponding standard deviation decreases from $26.12$ to $21.10$.

The extinction-time distributions in Figure~\ref{fig:pdf_extinction_time_scenario2} reveal a more moderate effect. The mean of $T$ decreases from $14.39$ weeks under delayed intervention to $12.99$ weeks under early intervention, whereas the standard deviation changes from $11.03$ to $11.79$ weeks. Thus, earlier intervention shortens the expected epidemic duration without reducing its dispersion, confirming that the timing of the switch affects $C^I$ and $T$ differently.

\begin{figure}[htbp]
    \centering
    \includegraphics[width=0.75\textwidth]{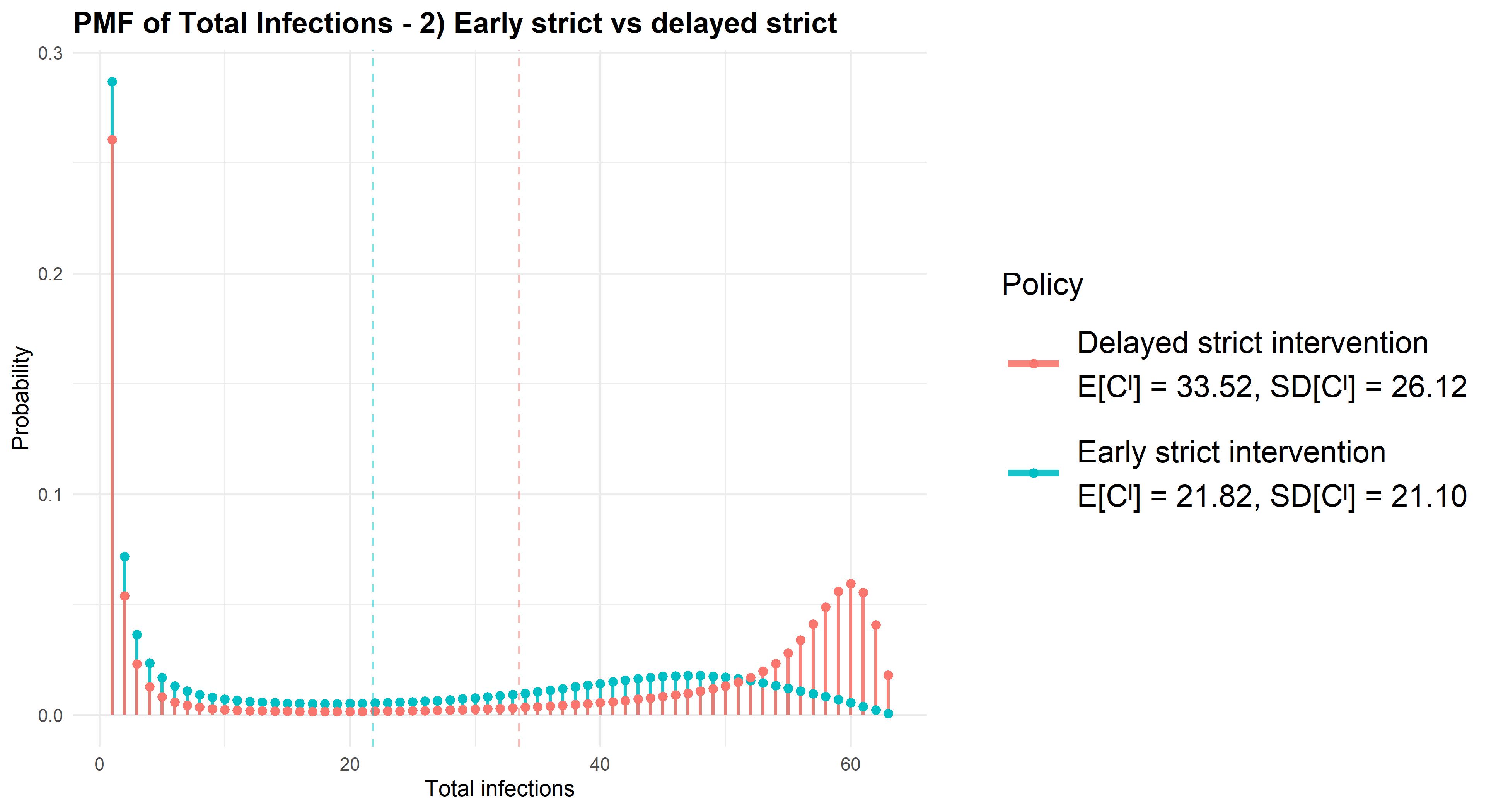}
    \caption{PMF of the total number of infected individuals under the early and delayed strict-intervention regimes considered in Scenario~2. Dashed vertical lines indicate the corresponding means.}
    \label{fig:pmf_total_infections_scenario2}
\end{figure}

\begin{figure}[htbp]
    \centering
    \includegraphics[width=0.75\textwidth]{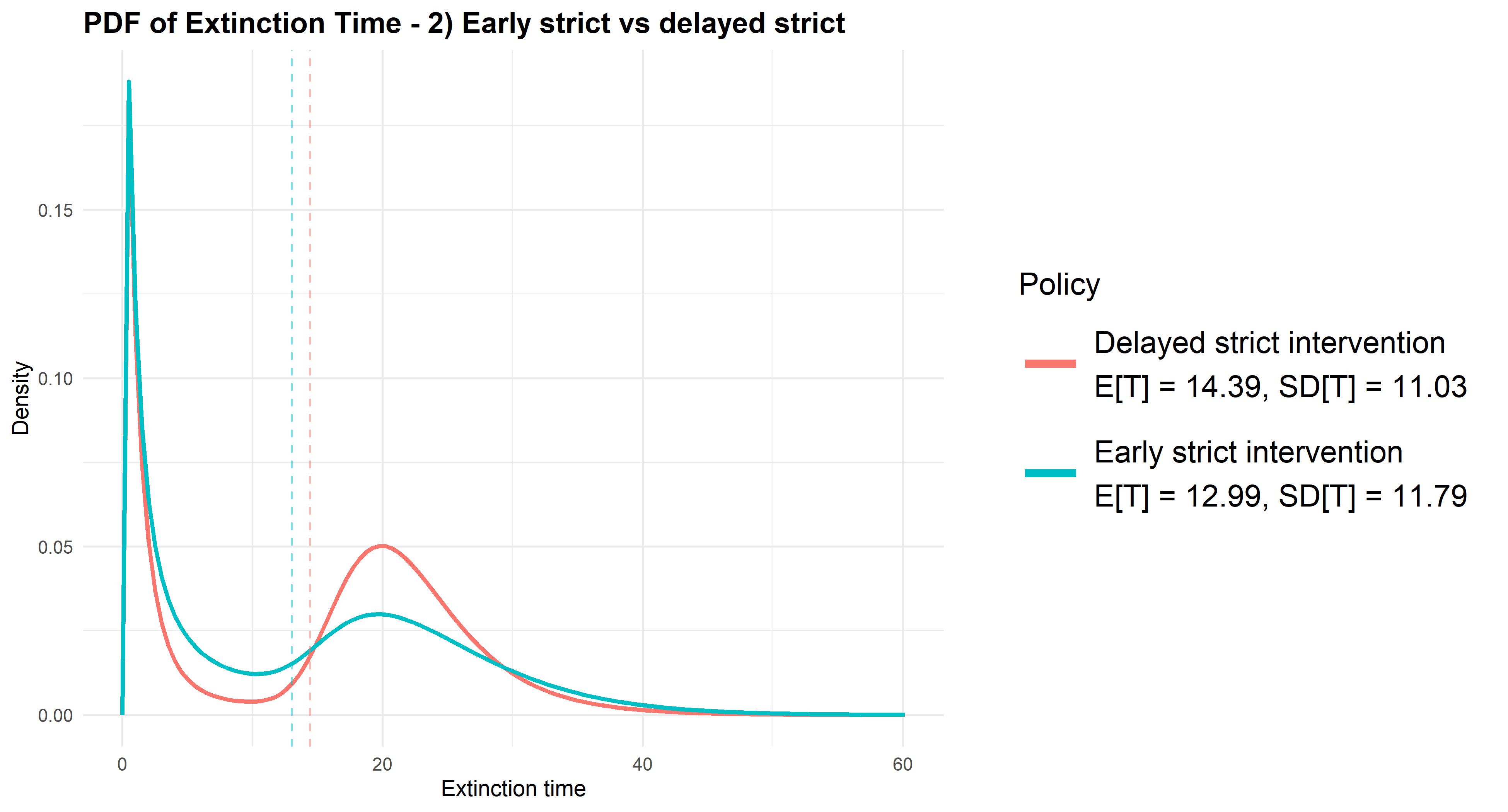}
    \caption{Density of the extinction time under the early and delayed strict-intervention regimes considered in Scenario~2. Time is measured in weeks, and dashed vertical lines indicate the corresponding means.}
    \label{fig:pdf_extinction_time_scenario2}
\end{figure}

\FloatBarrier
\subsection{Scenario 3: strict control and vaccination-supported control}
\label{sec3.4}

The third scenario compares strict control alone with strict control combined with direct immunity acquisition. There is no phase switching. The transmission and recovery parameters are
\begin{equation*}
b=0.495,
\qquad
\gamma=1/3.
\end{equation*}
The strict-control-only regime has $\psi=0$, whereas the vaccination-supported regime has $\psi=0.05$. This direct transition removes susceptible individuals from the infection process, as specified in Section~\ref{subsec:model}. The strict-control-only specification is identical to the strict-measures regime in Scenario~1 and is repeated here solely to isolate the additional effect of direct immunity acquisition.
Figure~\ref{fig:pmf_total_infections_scenario3} shows that vaccination support further reduces the infection count under strict control. The mean of $C^I$ decreases from $13.07$ to $6.08$, and the standard deviation decreases from $17.08$ to $7.55$.

Figure~\ref{fig:pdf_extinction_time_scenario3} shows the corresponding effect on extinction time. The mean of $T$ decreases from $10.19$ to $7.09$ weeks, and the standard deviation decreases from $11.86$ to $7.66$. In this scenario, adding direct immunity acquisition reduces both quantities and their dispersions.

\begin{figure}[htbp]
    \centering
    \includegraphics[width=0.80\textwidth]{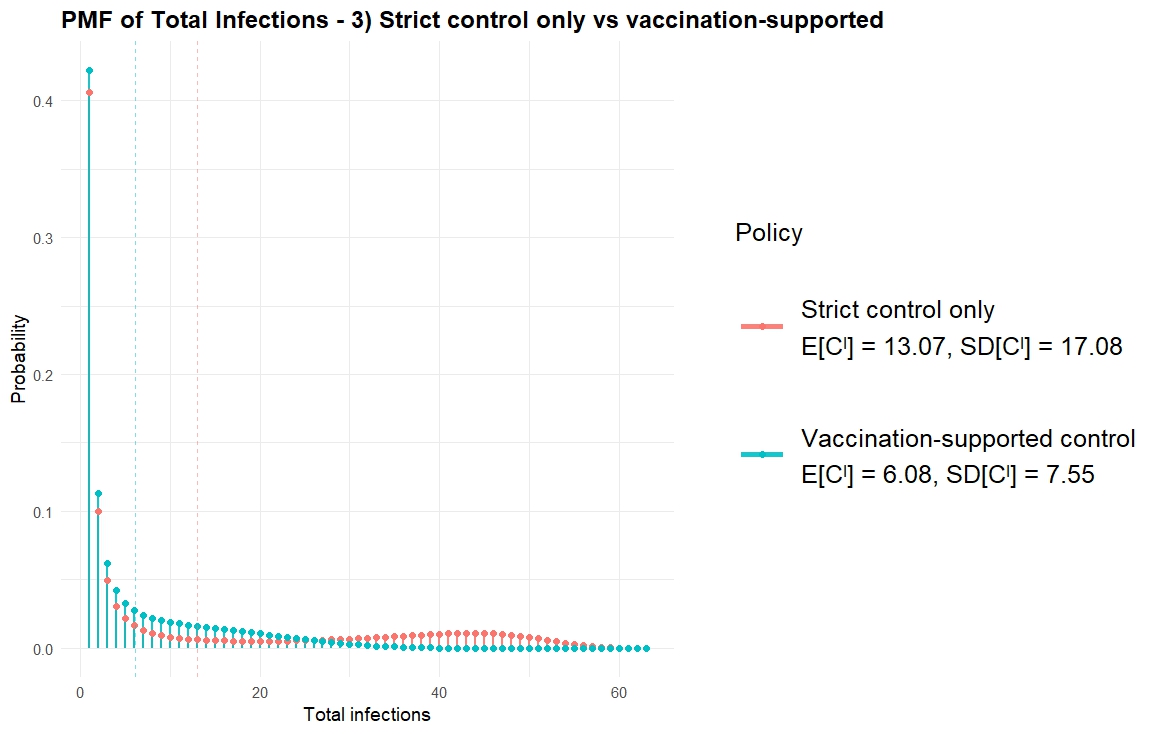}
    \caption{PMF of the total number of infected individuals under the strict-control-only and vaccination-supported regimes considered in Scenario~3. Dashed vertical lines indicate the corresponding means.}
    \label{fig:pmf_total_infections_scenario3}
\end{figure}

\begin{figure}[htbp]
    \centering
    \includegraphics[width=0.80\textwidth]{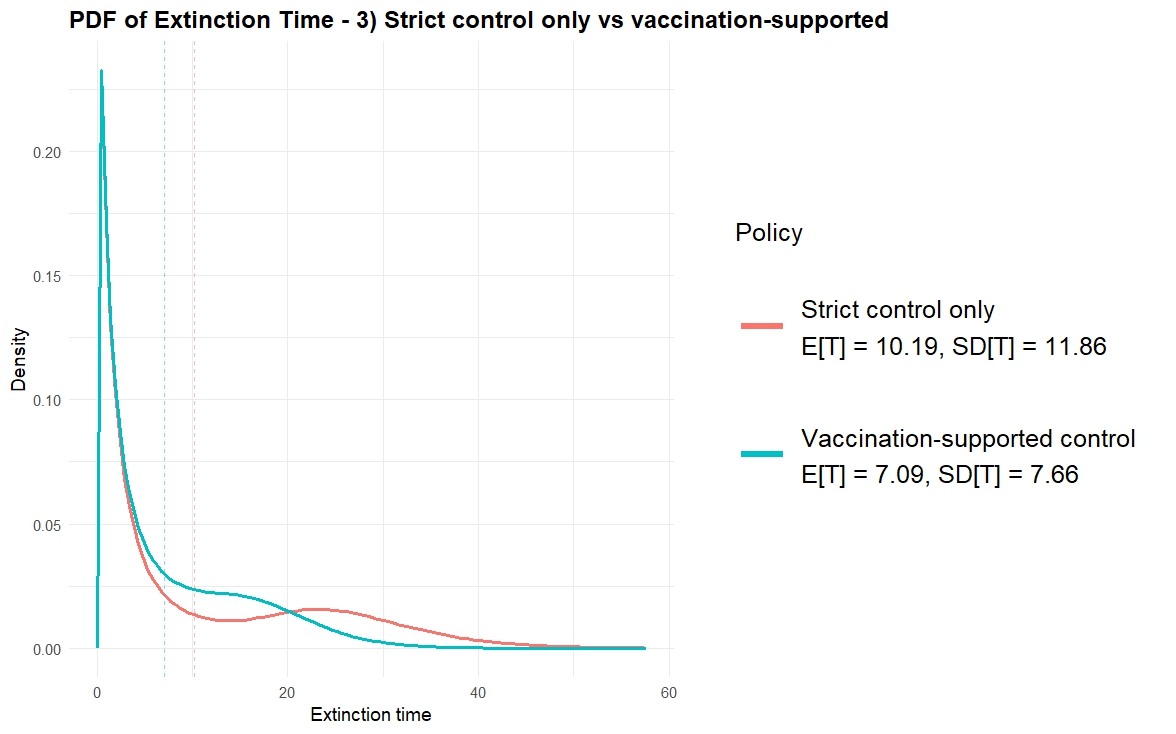}
    \caption{Density of the extinction time under the strict-control-only and vaccination-supported regimes considered in Scenario~3. Time is measured in weeks, and dashed vertical lines indicate the corresponding means.}
    \label{fig:pdf_extinction_time_scenario3}
\end{figure}

\FloatBarrier
\subsection{Scenario 4: state-dependent escalation}
\label{sec3.5}

The previous scenarios use state-independent switching. We now consider the state-dependent mechanism of Section~\ref{subsec:model}. The process starts in the no-measures phase and may move to the strict phase. Relaxation is excluded, so that $\lambda_{2,1}(s,i)=0$. The constant delayed-intervention case is used as reference:
\begin{equation*}
\lambda^{\mathrm{c}}_{1,2}(s,i)=0.05.
\end{equation*}
The state-dependent escalation rate is
\begin{equation}
\label{eq:state-dependent-scenario-rate}
\lambda^{\mathrm{sd}}_{1,2}(s,i)=0.01+0.40\frac{i^2}{i^2+6^2},
\qquad 1\leq i\leq N-s.
\end{equation}
Thus, escalation is unlikely when the number of infectious individuals is small and becomes more likely as $i$ increases. The epidemic parameters are the same as in Scenario~2:
\begin{equation*}
b_1=0.99,
\qquad
b_2=0.495,
\qquad
\gamma_1=\gamma_2=1/3,
\qquad
\psi_1=\psi_2=0.
\end{equation*}
Figure~\ref{fig:state_dependent_rate} shows the function in Eq.~\eqref{eq:state-dependent-scenario-rate}. Figures~\ref{fig:pmf_total_infections_scenario4} and~\ref{fig:pdf_extinction_time_scenario4} show the corresponding distributions of $C^I$ and $T$, respectively.

Compared with the constant delayed rate, the state-dependent rate decreases the mean of $C^I$ from $33.52$ to $28.50$, and the standard deviation from $26.12$ to $22.57$. The mean extinction times are close: $14.39$ weeks for the constant rate and $14.44$ weeks for the state-dependent rate. The corresponding standard deviations are $11.03$ and $11.37$. In this example, state-dependent escalation reduces the infection count by increasing the probability of strict control when the infectious population becomes large, while leaving the mean extinction time almost unchanged.

\begin{figure}[htbp]
    \centering
    \includegraphics[width=0.7\textwidth]{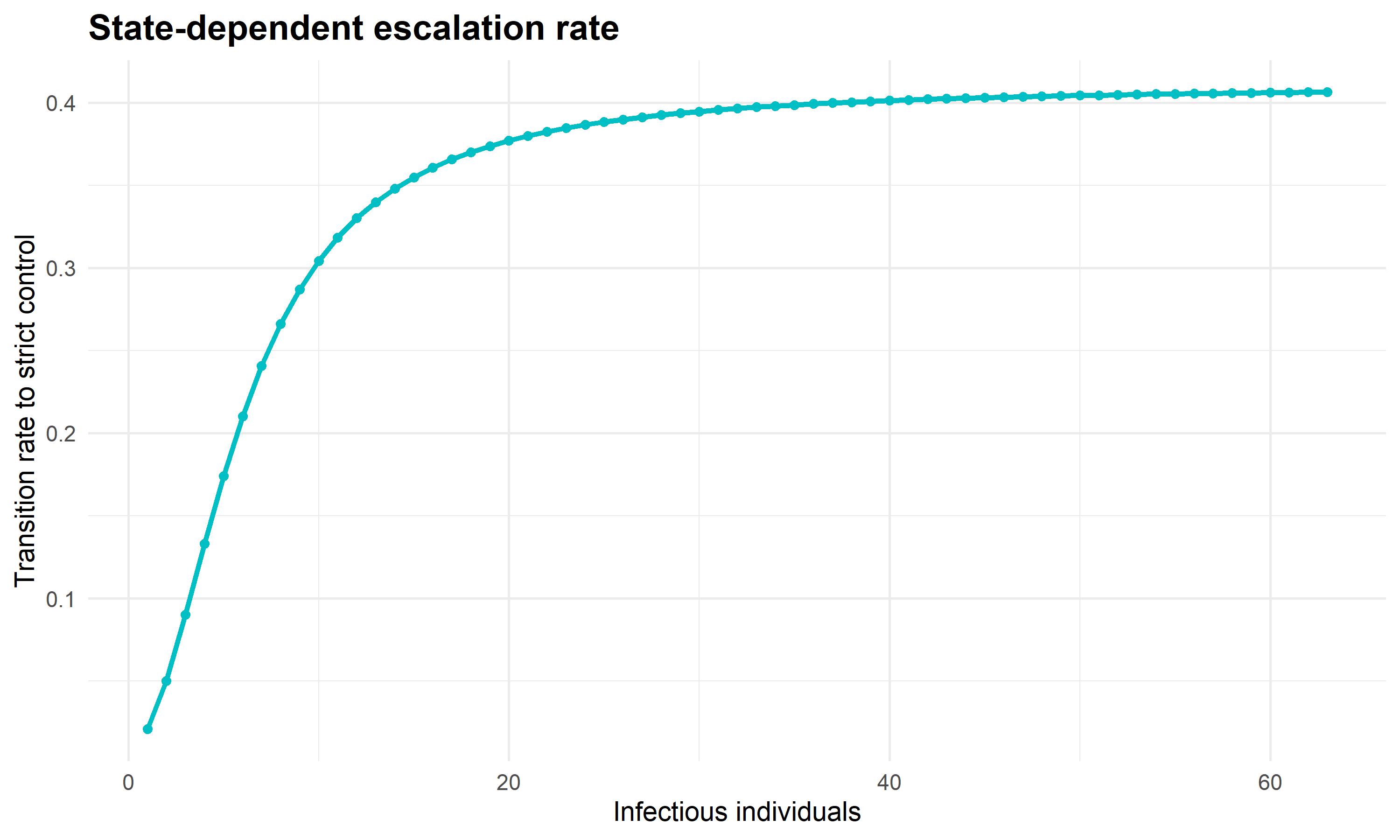}
    \caption{State-dependent escalation rate $\lambda^{\mathrm{sd}}_{1,2}(s,i)$ used in Scenario~4. Rates are expressed in $\mathrm{week}^{-1}$.}
    \label{fig:state_dependent_rate}
\end{figure}

\begin{figure}[htbp]
    \centering
    \includegraphics[width=0.8\textwidth]{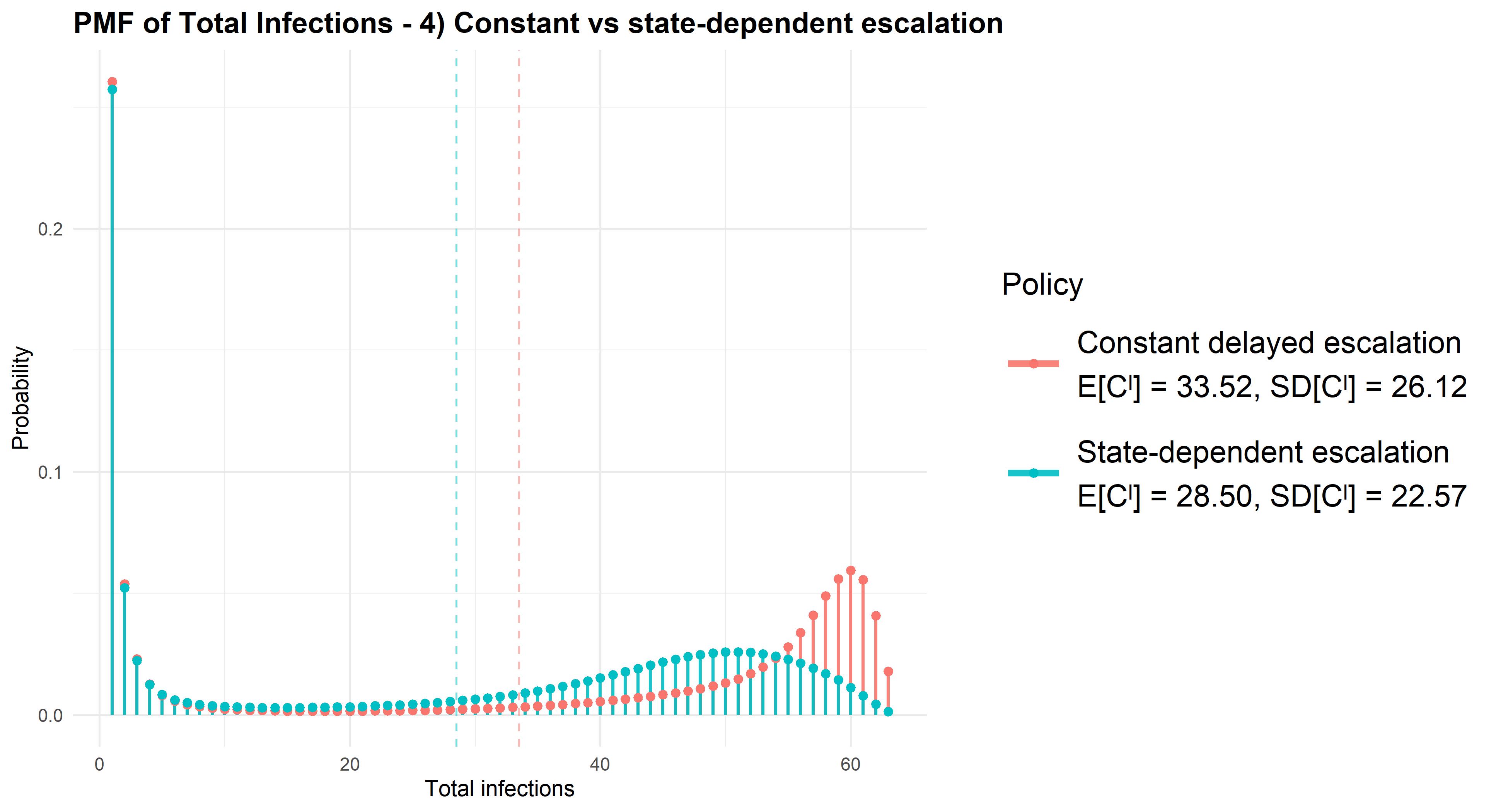}
    \caption{PMF of the total number of infected individuals under constant delayed escalation and state-dependent escalation. Dashed vertical lines indicate the corresponding means.}
    \label{fig:pmf_total_infections_scenario4}
\end{figure}

\begin{figure}[htbp]
    \centering
    \includegraphics[width=0.8\textwidth]{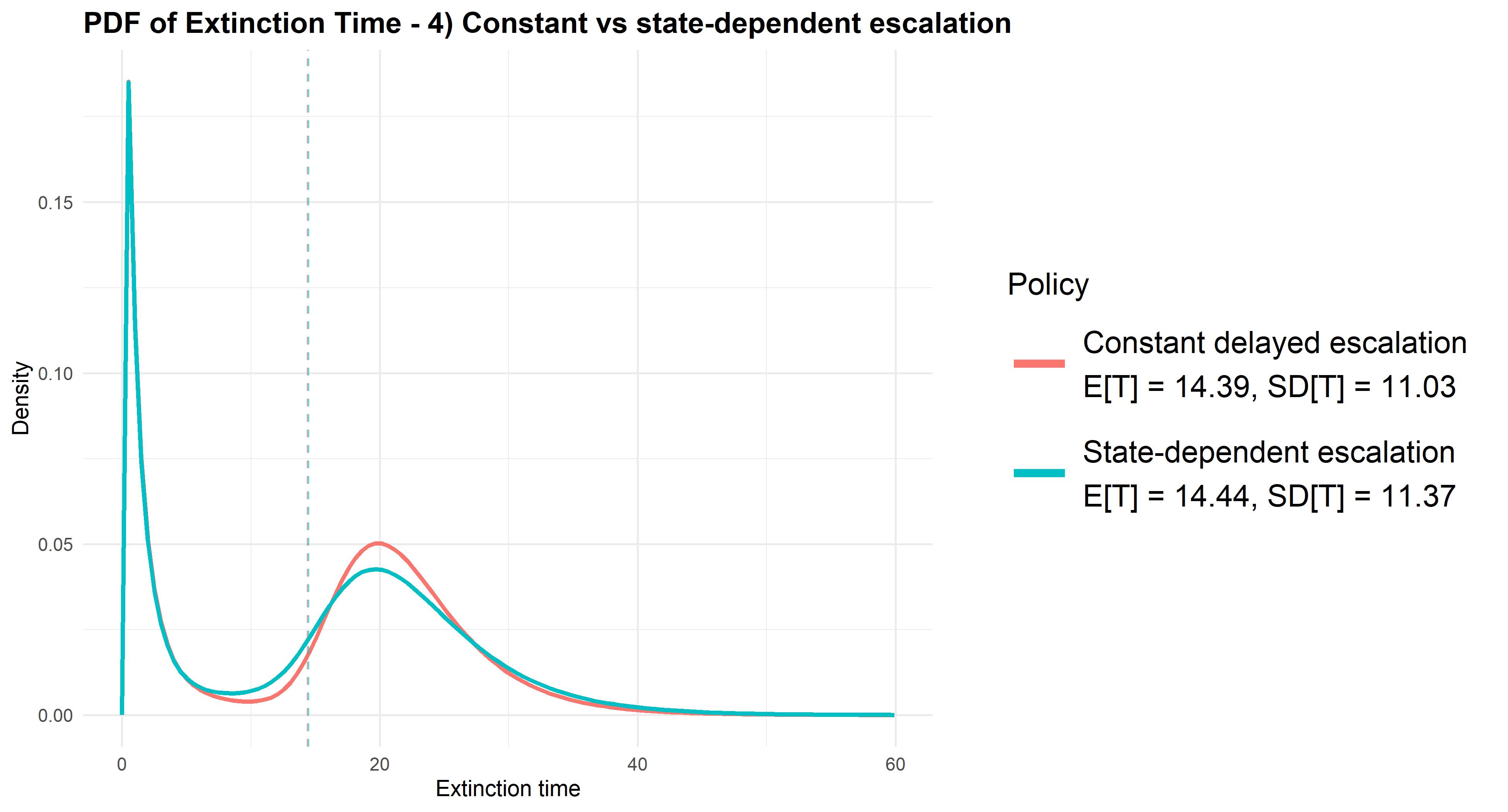}
    \caption{Density of the extinction time under constant delayed escalation and state-dependent escalation. Time is measured in weeks, and dashed vertical lines indicate the corresponding means.}
    \label{fig:pdf_extinction_time_scenario4}
\end{figure}

\FloatBarrier
\subsection{Cross-scenario synthesis}
\label{sec3.6}

Taken together, the four scenarios separate the effects of control intensity, intervention timing, direct immunity acquisition, and state-dependent escalation. Scenario~1 isolates transmission reduction; Scenario~2 compares early and delayed state-independent escalation; Scenario~3 assesses the additional effect of direct immunity acquisition; and Scenario~4 allows the escalation intensity to respond to the current infectious population. Across these mechanisms, changes in $C^I$ and $T$ need not occur in parallel: an intervention may substantially reduce infection burden while producing only a modest change in the mean extinction time or its variability. This cross-scenario comparison therefore demonstrates the value of examining epidemic burden and duration jointly rather than relying on a single summary measure.

\paragraph{Dependence between outbreak duration and size.}
The mixed moments also yield the correlation $\operatorname{Corr}(T,C^I)$ for every intervention specification. In Scenario~1, the correlations are 0.92, 0.93, and 0.93 under no measures, mild measures, and strict measures, respectively. In Scenario~2, they are 0.86 under early intervention and 0.88 under delayed intervention. In Scenario~3, the corresponding values are 0.93 for strict control alone and 0.85 for vaccination-supported control, while in Scenario~4 they are 0.88 for constant delayed escalation and 0.87 for state-dependent escalation. These values indicate a very strong and positive association between epidemic duration and infection burden. Reporting these correlations complements the marginal distributions by showing how the mechanisms alter the tendency of larger outbreaks to persist for longer periods.

\section{Discussion and conclusions} \label{sec4}
This paper develops a finite-population stochastic epidemic framework for outbreaks evolving under Markovian regime switching. The main novelty is twofold. First, the SIR epidemic process is augmented by an intervention phase, so that transmission, recovery, and direct immunity-acquisition rates may vary across regimes. The transition rates between regimes may also depend on the current epidemic state, allowing policy escalation to react probabilistically to the number of infectious individuals. This state dependence distinguishes the framework from conventional random-environment models in which regime changes are entirely exogenous to the epidemic trajectory. Second, the framework provides a unified joint distributional analysis of extinction time and infection burden, measured here through outbreak size, rather than focusing only on deterministic trajectories or mean outcomes.

The recursive formulas derived in Section~\ref{sec2} make it possible to obtain the infection-count distribution, extinction-time transforms and moments, conditional quantities, and mixed moments characterizing the dependence between outbreak duration and size. They therefore go beyond separate marginal summaries by quantifying the association between long-lasting and large outbreaks. These distributional quantities are exact up to the numerical solution of finite linear systems; only the continuous extinction-time densities displayed in Section~\ref{sec3} require numerical inversion of the Laplace--Stieltjes transform using the Abate--Whitt method. This is useful for comparing intervention mechanisms because different regimes may have similar expected durations but different dispersions or outbreak-size distributions.

The Luxembourg mpox case study demonstrates how the four specified intervention mechanisms can produce distinct changes in the distributions of infection burden and extinction time. The comparisons in Scenario~1 show that stronger transmission reduction shifts the infection burden towards smaller outbreak sizes and reduces the expected extinction time. The timing scenarios show that earlier strict intervention substantially reduces the expected number of infected individuals, while its effect on the variability of extinction time is less direct. The vaccination-supported scenario produces a further reduction in both infection burden and extinction time, illustrating how direct immunity acquisition can complement transmission-reducing control. Finally, the state-dependent escalation scenario shows that an intervention rate increasing with the current number of infectious individuals can reduce the infection count while leaving the mean extinction time almost unchanged. These results highlight why extinction time and infection burden should be analyzed jointly. The comparisons concern model-implied distributions under the specified mechanisms and should not be interpreted as causal estimates of the historical effects of interventions.

The empirical illustration is intentionally limited in scope. A baseline one-phase SIR model is calibrated to weekly incidence data, and the intervention scenarios are then evaluated conditionally on the fitted parameters and on specified regime mechanisms. This makes the role of each intervention assumption transparent, but it also means that the competing switching scenarios are not fitted directly to the data. A more detailed inferential analysis would need to assess the identifiability of regime-specific and switching parameters from aggregate incidence data and propagate both calibration and switching-parameter uncertainty into the reported distributional summaries. Finally, a minor limitation is that the first and last aggregated weeks in the mpox incidence series are partial weeks. They were retained to preserve the complete observation window, but this may introduce a small edge effect in the baseline calibration. 

Future work could extend the framework in several directions. More broadly, \citet{Hem2016} illustrated how community surveillance data and mathematical modeling can be combined to quantify an otherwise under-recognized disease burden. An important next step for the present framework would be to calibrate the full regime-switching model to richer real-world epidemic data, together with documented intervention histories, so that competing intervention scenarios can be compared empirically. Methodologically, the switching mechanism could depend on partially observed indicators, noisy surveillance data, or posterior beliefs about the epidemic state. A semi-Markov extension would also relax the exponential holding-time assumption for intervention regimes. Epidemiologically, the model could be extended to include latency, presymptomatic transmission, heterogeneous contact structure, waning immunity, or richer vaccination effects. For larger populations and richer state spaces, sparse implementations and scalable approximation methods would also be valuable. Such extensions would broaden the applicability of the framework while preserving its central aim: distributional stochastic evaluation of finite epidemics under changing intervention regimes.


\appendix

\section{Proofs for Section~\ref{sec2}}
\label{app:proofs}

\begin{proof}[Proof of Lemma~\ref{lem:absorption}]
Let $\Lambda=\max_{(s,i,p)\in S}q_{(s,i,p)}<\infty$. Uniformization by a Poisson process of rate $\Lambda$ shows that only finitely many jumps occur on each bounded time interval; hence the chain is non-explosive.

Fix $t_0>0$ and $(s,i,p)\in S_T$. Consider the event on which the next $i$ jumps are recoveries and each corresponding holding time is at most $t_0/N$. On this event the process follows
\begin{equation*}
(s,i,p),(s,i-1,p),\ldots,(s,1,p),(s,0,p)
\end{equation*}
and enters $S_a$ no later than $t_0$. At state $(s,j,p)$, $1\leq j\leq i$, the probability that the next jump is a recovery and occurs within $t_0/N$ is
\begin{equation*}
\int_0^{t_0/N}\gamma_p\,j\exp\{-q_{(s,j,p)}t\}\,\mathrm{d}t>0.
\end{equation*}
Successive conditioning gives a positive lower bound for $\mathbb{P}(T\leq t_0\,\mid\,(s,i,p))$. Since $S_T$ is finite, there exists $\delta\in(0,1)$ such that
\begin{equation*}
\inf_{(s,i,p)\in S_T}\mathbb{P}(T\leq t_0\,\mid\,(s,i,p))\geq\delta.
\end{equation*}
For $k\geq0$, the Markov property at time $k t_0$ gives
\begin{equation*}
\mathbb{P}_x(T>(k+1)t_0)
=\mathbb{E}_x\left[\mathbf{1}_{\{T>k t_0\}}\mathbb{P}_{X(k t_0)}(T>t_0)\right]
\leq(1-\delta)\mathbb{P}_x(T>k t_0),
\end{equation*}
for every $x\in S_T$. Hence
\begin{equation*}
\sup_{x\in S_T}\mathbb{P}_x(T>k t_0)\leq(1-\delta)^k,
\qquad k\geq0.
\end{equation*}
This proves that $T<\infty$ almost surely. With $c=-t_0^{-1}\log(1-\delta)>0$, the preceding inequality implies
\begin{equation*}
\sup_{x\in S_T}\mathbb{P}_x(T>t)\leq \exp\{c t_0\}\exp\{-ct\},\qquad t\geq0.
\end{equation*}
For $0<\eta<c$, Tonelli's theorem applied to $\exp\{\eta T\}-1=\int_0^T\eta\exp\{\eta t\}\,\mathrm{d}t$ yields
\begin{equation*}
\sup_{x\in S_T}\mathbb{E}_x[\exp\{\eta T\}]
\leq 1+\eta\exp\{c t_0\}\int_0^\infty\exp\{-(c-\eta)t\}\,\mathrm{d}t<\infty.
\end{equation*}
Take $\eta_T=c$. Since $T^k\leq k!\eta^{-k}\exp\{\eta T\}$ for $\eta>0$, all positive integer moments are finite.
\end{proof}

\begin{proof}[Proof of Theorem~\ref{th:joint-transform}]
Fix $x=(s,i,p)\in S_T$. Let $H$ be the first jump time and $J=X(H)$. Since $q_{(s,i,p)}\geq\gamma_p\,i\geq\gamma_*>0$, $H<\infty$ almost surely. For $y\neq x$,
\begin{equation*}
\mathbb{P}(H\in\mathrm{d}t,J=y\mid X(0)=x)=q_{x,y}\exp\{-q_{(s,i,p)}t\}\,\mathrm{d}t.
\end{equation*}
If $J=y\in S_T$, the strong Markov property at $H$ gives $T=H+T_y$ and $N^I=\kappa(x,y)+N_y^I$, where $(T_y,N_y^I)$ has the law of $(T,N^I)$ for a process started from $y$. If $J\in S_a$, then $T=H$ and the last jump is a recovery. Therefore,
\begin{align*}
\Phi_x(z,u)
={}&\sum_{\substack{y\in S_T\\y\neq x}}\int_0^\infty \exp\{-zt\}u^{\kappa(x,y)}\Phi_y(z,u)q_{x,y}\exp\{-q_{(s,i,p)}t\}\,\mathrm{d}t\\
&+\sum_{a\in S_a}\int_0^\infty \exp\{-zt\}q_{x,a}\exp\{-q_{(s,i,p)}t\}\,\mathrm{d}t.
\end{align*}
Because $\operatorname{Re}(z)>-\eta_0$ and $q_{(s,i,p)}\geq\gamma_*>\eta_0$, each integral converges and
\begin{equation}
\label{eq:appendix-first-step}
(z+q_{(s,i,p)})\Phi_x(z,u)=\sum_{\substack{y\in S_T\\y\neq x}}q_{x,y}u^{\kappa(x,y)}\Phi_y(z,u)+\sum_{a\in S_a}q_{x,a}.
\end{equation}
Writing~\eqref{eq:appendix-first-step} for all $x\in S_T$ gives~\eqref{eq:global-joint-system}.

For $s\geq1$, Eq.~\eqref{eq:appendix-first-step} becomes
\begin{align*}
(z+q_{(s,i,p)})\Phi_{s,i,p}(z,u)
={}&\gamma_p\,i\,\Phi_{s,i-1,p}(z,u)+\psi_p\,s\,\Phi_{s-1,i,p}(z,u)\\
&+u\,\dfrac{b_p}{N}\,s\,i\,\Phi_{s-1,i+1,p}(z,u)+\sum_{\substack{p'=1\\p'\neq p}}^{P}\lambda_{p,p'}(s,i)\Phi_{s,i,p'}(z,u).
\end{align*}
For $i=1$, the first term contains $\Phi_{s,0,p}(z,u)=1$. For $s=0$, the infection and immunity-acquisition terms vanish. Collecting the scalar equations over $p=1,\ldots,P$ gives~\eqref{eq:joint-recursion-zero} and~\eqref{eq:joint-recursion}.

It remains to prove nonsingularity. Let $a_{p,p}$ be the diagonal entry of $\mathbf{A}_{s,i}(z)$. Then
\begin{equation*}
a_{p,p}=z+\dfrac{b_p}{N}\,s\,i+\gamma_p\,i+\psi_p\,s+\sum_{\substack{p'=1\\p'\neq p}}^{P}\lambda_{p,p'}(s,i),
\end{equation*}
and the off-diagonal entries in row $p$ are $-\lambda_{p,p'}(s,i)$. Since $i\geq1$ and $\operatorname{Re}(z)>-\eta_0$,
\begin{equation*}
|a_{p,p}|\geq\operatorname{Re}(a_{p,p})>\sum_{\substack{p'=1\\p'\neq p}}^{P}\lambda_{p,p'}(s,i).
\end{equation*}
If $\mathbf{A}_{s,i}(z)\mathbf{v}=\mathbf{0}$ and $|v_p|=\max_r|v_r|$, the $p$th row gives
\begin{equation*}
|a_{p,p}|\,|v_p|\leq\left(\sum_{\substack{p'=1\\p'\neq p}}^{P}\lambda_{p,p'}(s,i)\right)|v_p|.
\end{equation*}
If $v_p\neq0$, the last inequality contradicts the strict inequality in the preceding display. Hence $v_p=0$ and, by maximality, $\mathbf{v}=\mathbf{0}$. Thus, $\mathbf{A}_{s,i}(z)$ is nonsingular.

The recursion is now well defined. At level $s=0$, \eqref{eq:joint-recursion-zero} determines $\boldsymbol{\Phi}_{0,i}$ successively from $i=1$ to $N$, starting from $\boldsymbol{\Phi}_{0,0}=\mathbf{1}_P$. If all vectors at levels $0,\ldots,s-1$ are known, \eqref{eq:joint-recursion} determines the vectors at level $s$ successively in $i$, starting from $\boldsymbol{\Phi}_{s,0}=\mathbf{1}_P$. The solution of the recursive system is unique and satisfies the global system, which proves the theorem.
\end{proof}

\begin{proof}[Proof of Corollary~\ref{cor:coefficient-recursion}]
If $i=0$, then $T=0$ and $N^I=0$, which gives the boundary values. If $s=0$, no infection can occur; hence $\mathbf{h}_{0,i}^{(n)}(z)=\mathbf{0}_{P\times1}$ for $n\neq0$, and the coefficient of $u^0$ in~\eqref{eq:joint-recursion-zero} gives~\eqref{eq:coefficient-zero}.

For $s\geq1$, substitute~\eqref{eq:joint-polynomial-expansion} into~\eqref{eq:joint-recursion}. The recovery and immunity-acquisition terms preserve the power of $u$, whereas
\begin{equation*}
u\sum_{m=0}^{s-1}u^m\mathbf{h}_{s-1,i+1}^{(m)}(z)=\sum_{n=1}^{s}u^n\mathbf{h}_{s-1,i+1}^{(n-1)}(z).
\end{equation*}
Equality of polynomial coefficients gives~\eqref{eq:coefficient-recursion}. Nonsingularity of $\mathbf{A}_{s,i}(z)$ and induction in $s$ and $i$ give uniqueness.
\end{proof}

\begin{proof}[Proof of the conditional identities~\eqref{eq:conditional-lst}--\eqref{eq:conditional-moments}]
By conditional expectation,
\begin{equation*}
\mathbb{E}\left[\exp\{-zT\}\,\middle|\,N^I=n,(s,i,p)\right]
=\dfrac{\mathbb{E}\left[\exp\{-zT\}\mathbf{1}_{\{N^I=n\}}\,\middle|\,(s,i,p)\right]}{\mathbb{P}(N^I=n\,\mid\,(s,i,p))},
\end{equation*}
which gives~\eqref{eq:conditional-lst}. To differentiate $h_{s,i,p}^{(n)}$, choose a compact set in $\{z:\operatorname{Re}(z)>-\eta_0\}$ and then $\rho<\eta<\eta_T$ so that $-\operatorname{Re}(z)\leq\rho$ on a neighbourhood of the compact set. The bound
\begin{equation*}
T^k\exp\{\rho T\}\leq \dfrac{k!}{(\eta-\rho)^k}\exp\{\eta T\}
\end{equation*}
and Lemma~\ref{lem:absorption} justify differentiation under the expectation. Hence
\begin{equation*}
\dfrac{\partial^k}{\partial z^k}h_{s,i,p}^{(n)}(z)=(-1)^k\mathbb{E}\left[T^k\exp\{-zT\}\mathbf{1}_{\{N^I=n\}}\,\middle|\,(s,i,p)\right].
\end{equation*}
Setting $z=0$ and dividing by $h_{s,i,p}^{(n)}(0)>0$ gives~\eqref{eq:conditional-moments}.
\end{proof}

\begin{proof}[Proof of Theorem~\ref{th:mixed-moments}]
The boundary values follow from $T=0$ and $N^I=0$ when $i=0$. For a transient state, $T^0(N^I)_0=1$.

The differentiation used below is justified by the same domination argument as in the preceding proof, together with the fact that $N^I\leq s$. Define
\begin{equation*}
\mathcal{D}_{k,r}f=(-1)^k\left.\partial_z^k\partial_u^r f(z,u)\right|_{z=0,u=1}.
\end{equation*}
Then $\mathcal{D}_{k,r}\boldsymbol{\Phi}_{s,i}=\boldsymbol{\mu}_{s,i}^{(k,r)}$.

Apply $\mathcal{D}_{k,r}$ to~\eqref{eq:joint-recursion-zero}. Since $\mathbf{A}_{0,i}(z)=\mathbf{A}_{0,i}(0)+z\mathbf{I}_P$, Leibniz' rule gives $\mathbf{A}_{0,i}(0)\boldsymbol{\mu}_{0,i}^{(k,r)}-k\boldsymbol{\mu}_{0,i}^{(k-1,r)}$ on the left-hand side, while the right-hand side gives $i\,\boldsymbol{\Gamma}\boldsymbol{\mu}_{0,i-1}^{(k,r)}$. This proves~\eqref{eq:mixed-recursion-zero}.

Apply $\mathcal{D}_{k,r}$ to~\eqref{eq:joint-recursion}. The recovery and immunity-acquisition terms give the corresponding two terms in~\eqref{eq:mixed-recursion}. For the infection term, evaluated at $(z,u)=(0,1)$, use
\begin{equation*}
\partial_u^r\left(u\boldsymbol{\Phi}_{s-1,i+1}\right)
=\partial_u^r\boldsymbol{\Phi}_{s-1,i+1}
+r\partial_u^{r-1}\boldsymbol{\Phi}_{s-1,i+1},
\end{equation*}
where the second term is omitted when $r=0$. After multiplication by $s\,i\,\mathbf{B}/N$, this gives the last term in~\eqref{eq:mixed-recursion}; hence~\eqref{eq:mixed-recursion} follows.

For uniqueness, proceed by induction on $d=k+r$. The case $d=0$ is known. Suppose all moments of total order less than $d$ have been computed and fix $(k,r)$ with $k+r=d$. At level $s=0$, \eqref{eq:mixed-recursion-zero} determines the vectors successively in $i$, because it contains a vector of total order $d-1$ and a vector of total order $d$ with infectious index $i-1$. If all vectors of order $d$ are known at levels below $s$, then \eqref{eq:mixed-recursion} determines the vectors at level $s$ successively in $i$. The matrices $\mathbf{A}_{s,i}(0)$ are nonsingular by Theorem~\ref{th:joint-transform}. This proves uniqueness.
\end{proof}

\section*{Data and code availability}
The empirical application uses publicly available mpox incidence data from the Our World in Data mpox repository \citep{EdouardMathieu2022}. The processed data and computer code used to reproduce the numerical results are available from the authors upon reasonable request.

\section*{Declaration of competing interests}
The authors declare that they have no known competing financial interests or personal relationships that could have appeared to influence the work reported in this paper.

\section*{Funding and acknowledgements}
This work was partially supported by the French National Research Agency (ANR) under grant ANR-21-CE40-0005 (HSMM-INCA). The authors gratefully acknowledge LIEC, Université de Lorraine, France, which hosted the first author as an invited researcher during part of this work. The authors also thank Michel Vaillant and the CCMS team at the Luxembourg Institute of Health (LIH) for fruitful discussions.

\bibliographystyle{apalike}
\bibliography{export_cited_alphabetical}

\end{document}